\documentclass[aps,pra,twocolumn,superscriptaddress,showpacs,final,floatfix]{revtex4-1}

\usepackage{amsmath,amsfonts,amsthm,amssymb}
\usepackage{graphicx}
\usepackage[utf8]{inputenc}
\usepackage[sc,osf]{mathpazo}
\usepackage{braket}
\usepackage{setspace}
\usepackage[colorlinks=true,citecolor=blue,urlcolor=magenta]{hyperref}
\usepackage{bbm}
\usepackage{color}

\usepackage{tikz}
\usepackage{tikzscale}
\usepackage{standalone}
\usetikzlibrary{arrows}
\usetikzlibrary{positioning}
\usetikzlibrary{patterns}
\usepackage{lipsum}
\usepackage{cancel}

\newtheorem{proposition?}{Proposition?}

\newtheorem{lemma}{Lemma}
\newtheorem{obs}{Observation}
\newtheorem{rem}{Remark}

\theoremstyle{definition}

\newcommand{\tr}[1]{\mathrm{tr}\left[#1\right]} 
\newcommand{\ketbra}[1]{\ensuremath{| #1 \rangle \!\langle #1 |}}
\newcommand{\mean}[1]{\langle #1\rangle}

\newcommand{\id}{{\rm Id}} 

\newcommand{\A}{\mathsf{A}}
\newcommand{\B}{\mathsf{B}}
\newcommand{\C}{\mathsf{C}}
\newcommand{\E}{\mathsf{E}}
\newcommand{\OO}{\mathsf{O}}
\newcommand{\Go}{\mathsf{G}}

\newcommand{\Q}{\mathsf{Q}}
\newcommand{\W}{\mathsf{W}}
\newcommand{\td}{\widetilde}

\newcommand{\II}{\mathcal I}
\newcommand{\EE}{\mathcal E}
\newcommand{\UU}{\mathcal U}
\newcommand\ND{\mathrel{\stackrel{\makebox[0pt]{\mbox{\normalfont\tiny ND}}}{\rightarrow}}}
\newcommand\NND{\mathrel{\stackrel{\makebox[0pt]{\mbox{\normalfont\tiny ND}}}{\leftrightarrow}}}

\begin{document}

\title{Leggett-Garg macrorealism and the quantum nondisturbance conditions}
\author{Roope Uola}
\thanks{All authors contributed equally.}
\affiliation{Naturwissenschaftlich-Technische 
Fakult\"at, Universit\"at Siegen, Walter-Flex-Strasse 3, 57068 Siegen, Germany}
\affiliation{D{\'e}partement de Physique Appliqu{\'e}e, Universit{\'e} de Gen{\`e}ve, CH-1211 Gen{\`e}ve, Switzerland}
\email{roope.uola@gmail.com}

\author{Giuseppe Vitagliano}
\thanks{All authors contributed equally.}
\affiliation{Institute for Quantum Optics and Quantum Information (IQOQI), Austrian Academy of Sciences, Boltzmanngasse 3, 1090 Vienna, Austria}
\email{giuseppe.vitagliano@oeaw.ac.at}\email{costantino.budroni@oeaw.ac.at}

\author{Costantino Budroni}
\thanks{All authors contributed equally.}
\affiliation{Institute for Quantum Optics and Quantum Information (IQOQI), Austrian Academy of Sciences, Boltzmanngasse 3, 1090 Vienna, Austria}
\affiliation{Faculty of Physics, University of Vienna, Boltzmanngasse 5, 1090 Vienna, Austria}
\date{11th~July~2019}. 

\begin{abstract}  
We investigate the relation between a refined version of Leggett and Garg conditions for macrorealism, namely the 
no-signaling-in-time (NSIT) conditions, and the quantum mechanical notion of nondisturbance between measurements. We show 
that all NSIT conditions are satisfied for any state preparation if and only if certain compatibility criteria on the 
state-update rules relative to the measurements, i.e. quantum instruments, are met. The systematic treatment of NSIT 
conditions supported by structural results on nondisturbance provides a unified framework for the 
discussion of the the clumsiness loophole. This extends previous approaches and allows for a tightening of the loophole via 
a hierarchy of tests able to disprove a larger class of macrorealist theories. Finally, we discuss perspectives 
for a resource theory of quantum mechanical disturbance related to violations of macrorealism.

\end{abstract}
\maketitle

\section{Introduction}

The notion of {\it macrorealism} can be traced back to the intuition that assigning a definite value to a macroscopic variable, e.g. the position of the moon, should be possible at any time and that, in principle, one should be able to observe such a value with negligible disturbance. A paradigmatic example showing the counterintuitive consequences of violating these assumptions is the famous thought experiment by 
Schr\"odinger involving a cat in a superposition of being dead and alive. Leggett and Garg \cite{LeggettPRL1985,EmaryRPP2014} formalized the assumptions as macrorealism {\it per se} and noninvasive measurability with the goal of testing these principle in the laboratory. For this purpose they derived conditions that observed statistics must satisfy. The results are the so-called Leggett-Garg inequalities (LGI), which similarly to Bell inequalities \cite{Bell1964} are able to put at test quantum mechanics versus the predictions of macrorealist theories.

In contrast to Bell scenarios in which laboratories are far apart, in a single system evolving in time signalling is possible in one direction, i.e., the direction of time flow. Connected with this fact, refinements of the original LGI have been proposed (cf., e.g., Refs.~\cite{Morikoshi2006, DeviPRA2013, KoflerPRA2013, LiSR2012, ClementePRA2015, ClementePRL2016}) leading to the so-called no-signaling-in-time (NSIT) conditions. Importantly, NSIT conditions provide necessary and sufficient conditions for an observed probability distribution to admit a macrorealist model~\cite{ClementePRL2016}, provided that the sequence of maximal length is measurable~\footnote{Contrary to Bell and contextuality tests, where the measurement one can jointly perform are limited to compatible measurements, in Leggett-Garg test a sequence of all possible measurements is, in principle, allowed. By definition of macrorealist models, then, the obtained probability distribution must correspond to the macrorealist hidden variable model. Hence, it is sufficient to compare the probability observed on shorter sequences, with the marginal of such a distribution. If further constraints are imposed on the joint measurement that can be performed in a Leggett-Garg test, e.g., due to practical limitations, then the NSIT are no longer necessary and sufficient conditions for macrorealism.}. 
The NSIT conditions have a direct interpretation in terms of disturbance introduced by the measurement apparatus on the system. For the simple case of a single observable measured at two time steps, say $\Q(t_1)$ evolving into $\Q(t_2)$, the NSIT condition is
satisfied if by monitoring $\Q(t_2)$ one cannot detect whether a measurement of $\Q(t_1)$ has been performed. In other words, $\Q(t_1)$ cannot ``disturb'' the statistics of $\Q(t_2)$.

From an experimental perspective, it is evident that the disturbance leading to a violation of NSIT may be classically explainable in terms of imperfections in the measurement apparatus. Thus, a major challenge opens in all practical Leggett-Garg tests, the {\it clumsiness loophole} \cite{WildeMizel2012}. Such loophole cannot be completely closed due to practical, i.e., unavoidable imperfections, as well as fundamental reasons, e.g., the Heisenberg principle in its instance about the fundamental disturbance associated with measurements of incompatible observables (cf Ref.~\cite{BLW_review}). So far, several experiments have been performed showing the violation of LGI, see, e.g., Ref.~\cite{EmaryRPP2014} for a review of older experiments and Refs.~\cite{RobensPRX2015, ZhouPRL2015, WangPRA2016, KneeNC2016, Formaggio2016, HuffmanPRA2017, KatiyarNJP2017, WangPRA2018} for more recent ones. Furthermore, ways of tightening the loophole have been proposed and implemented, which include special measurement implementations and additional tests involving different evolutions or different preparations \cite{LeggettPRL1985, WildeMizel2012, KneeNC2012, GeorgePNAS2013, LG_QNDPRL2015,  RobensPRX2015, KneeNC2016, EmaryPRA2017, HuffmanPRA2017}.

At this point, it is worth noticing that in the NSIT conditions, as well as in Bell inequalities, the concept of observable is specified only at the level of its outcome statistics. From a quantum mechanical perspective, the most general way to associate a probability distribution to a measurement is through a so-called {\it Positive Operator-Valued Measure} (POVM), that generalizes the definition of observable as an Hermitian operator. The latter can be seen as a {\it Projector Valued Measure} (PVM). The notion of POVM provides nontrivial compatibility properties and a richer structure than that allowed by PVMs.
In the case of PVMs, in fact, all definitions of compatibility reduce to commutativity \cite{HW2010,BuschBook,BuschBook2016}, while this is not the case for general POVMs. In particular, it is known that the compatibility notion of quantum nondisturbance is weaker than commutativity, i.e., examples of POVMs that are noncommuting and yet nondisturbing have been found \cite{HW2010}. Thus, it makes sense to consider the most general quantum observables, namely POVMs, for a systematic analysis of NSIT constraints in quantum mechanics, and for further tightening of the clumsiness loophole.

In order to minimize the assumptions to be made in practical tests, it is important to separate the unavoidable experimental imperfections, i.e., the ``clumsiness'' associated with a specific experimental realization, from the fundamental limitations predicted by quantum mechanics for measurements of incompatible observables. Following this approach, we show how a systematic analysis of NSIT constraints from the perspective of the quantum notion of disturbance allows to tighten the clumsiness loophole, in the same spirit as Wilde and Mizel~\cite{WildeMizel2012} and George {\it et al.}~\cite{GeorgePNAS2013}. In particular, the assumption of {\it noninvasive measurements} can be weakened to a generalized notion of {\it noncolluding measurements}~\cite{WildeMizel2012}. This allows one to design more refined experimental tests of macrorealism for which an explanation in terms of ``clumsy'' measurements is harder to maintain.

Our starting point is the formal notion of {\it nondisturbance} for pairs of quantum observables introduced in Ref.~\cite{HW2010} (see also 
Refs.~\cite{BuschBook,BuschBook2016}). This notion can be shown to correspond to the NSIT condition for two measurements when optimization over all states and measurement implementations is performed. Following this analogy, we define nondisturbance conditions 
for arbitrary sequences of measurements and prove that if certain 
minimal subset of them is satisfied, then it is possible to satisfy macrorealism with all state preparations. After that, we prove the existence of certain compatibility structures such as existence of pairwise mutually nondisturbing POVMs, which become disturbing when measured in a longer sequence. On the one hand, these structures allow one to design stronger tests of macrorealism as anticipated above. On the other hand, we show that these structures must be taken into account in order to define a meaningful measure for violations of macrorealism. Finally, we give an example of such a measure and study the operations that do not increase it.

The paper is organized as follows. Sec.~\ref{mainsec:II} contains a short conceptual review of macrorealism and NSIT conditions (Sec.~\ref{sec:II}) as well as quantum incompatibility (Sec.~\ref{sec:III}). The main results are stated in Sec.~\ref{sec:IV}. Namely, we define nondisturbance conditions for arbitrary sequences of measurements and associate them with NSIT conditions. Moreover, we investigate the structural aspects of quantum disturbance (Sec.~\ref{sec:V}) and discuss how our analysis unifies and extends previous approaches to the clumsiness loophole (Sec.~\ref{sec:clumsloop}). In Sec.~\ref{sec:resourceperspective}, we discuss a possible resource-theoretical extension of our approach by defining a measure of quantum disturbance connected with the violation of macrorealism and investigating the operations that do not increase it.  Finally in Sec.~\ref{sec:conclusion}, we present the conclusions and the outlook for near-future developments.

\section{Preliminary notions}\label{mainsec:II}

\subsection{Macrorealism and No-Signalling In Time}\label{sec:II}

Macrorealist theories are defined by two assumptions: macrorealism {\it per se} (MRps) and noninvasive measurability (NIM). Such assumptions provide conditions on the measurement outcome probabilities in a sequential measurement scenario. They can be cast as
\begin{itemize}
\item (MRps): It is possible to assign a definite value to $Q(t)$ at any time $t$
\item (NIM): The value of $Q(t)$ can be measured with arbitrary small disturbance of its subsequent evolution.
\end{itemize}
The first condition implies that the distribution of $Q(t)$ at different instants of time is given by a classical variable, whereas the second one implies that we can find a nondisturbing measurement implementation, such that the probability distribution does not change as a consequence of our measurement procedure. 

Let us explain this with a simple example (depicted in Fig.~\ref{fig:Fig1}). Consider a single qubit in an initial state $\ket{\Psi}$ and a sequence of two projective measurements $Q(t_1)=\sigma_z$ and its time evolved version $Q(t_2)=\exp(i\pi \sigma_y) \sigma_z \exp(-i\pi \sigma_y) = \sigma_x$. Using the projection postulate, we can compute the joint probability of outcomes $(q_1,q_2)$ as
\begin{equation}
p(q_1,q_2) = \bra{\Psi} \Pi^{q_1}_{\sigma_z} \Pi^{q_2}_{\sigma_x} \Pi^{q_1}_{\sigma_z} \ket{\Psi} ,
\end{equation}
where we indicated with $\{\Pi^{q}_{\sigma_{z/x}}\}_{q=\pm 1}$ the projectors onto the eigenstates of
$\sigma_{z/x}$ respectively. On the other hand, for a similar sequence in which the first measurement is not performed we have 
\begin{equation}
p(q_2) = \bra{\Psi} \Pi^{q_2}_{\sigma_x} \ket{\Psi} .
\end{equation}
For the initial state $\ket{\Psi}=\ket{1}_z$ (an eigenstate of $\sigma_z$) we have $p(q_2)=\sum_{q_1} p(q_1,q_2)$, i.e., we cannot detect whether a measurement of $\sigma_z$ was performed and the outcome discarded before the measurement of $\sigma_x$. With this initial state macrorealism is satisfied. However, if we choose an eigenstate of $\sigma_x$ as the initial state, i.e., $\ket{\Psi}=\ket{1}_x=\tfrac 1 {\sqrt 2}(\ket{1}_z+\ket{-1}_z)$ we have $\sum_{q_1} p(q_1,q_2)= 1/2$ for both outcomes $q_2=\pm 1$ while $p(q_2)=\delta_{q_2 1}$.
Hence, for the second choice of initial state (and for the same sequence) macrorealism is not satisfied when using the projection postulate. 
Note also that, even with the first choice of initial state $\ket{\Psi}=\ket{1}_z$, one can observe the violation of macrorealism by choosing a different (e.g. noisy) implementation of the first measurement. Alternatively, using the projection postulate followed by a measure-and-prepare channel (preparing always the state $|1\rangle_x$ on every run) observations are consistent with macrorealism even when the initial state is $|1\rangle_x$. With this, we emphasize again that the notion of macrorealism is deeply connected with preparations as well as with measurement implementations.

\begin{figure}[t]
\includegraphics[width=0.45\textwidth]{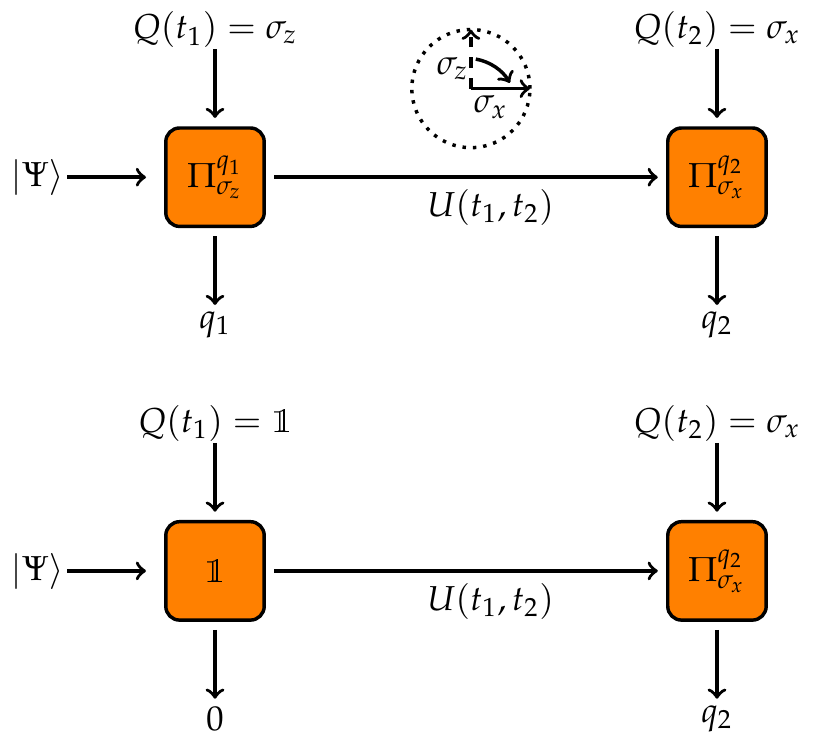}
\caption{Two simple sequences of two measurements that lead to a violation of macrorealism: From an initial state $\ket{\Psi}$ either a projective measurement of $\sigma_z$ or no measurement is performed, followed by a unitary rotation $\sigma_z \mapsto \sigma_x$ and a projective measurement of $\sigma_x$. For an initial superposition state $\ket{\Psi}=\tfrac 1 {\sqrt 2}(\ket{1}_z+\ket{-1}_z)$ violation of macrorealism can be observed.}\label{fig:Fig1}
\end{figure}

To simplify the subsequent discussion, let us fix some notation. We consider  probability distributions $p(q_1,\ldots,q_n|s_1,\ldots,s_n)$, where $q_i$ denotes the outcome at time $t_i$ and $s_i$ denotes the corresponding measurement setting.
In this work, we restrict to the case of two measurement settings. We label the event ``no measurement'' as $0$ (and a fixed outcome $q=0$) and the event ``measurement'' as $1$.
On the one hand, this simplifies the notation and the discussion of nondisturbance conditions. On the other hand, our construction can be straightforwardly generalized to the case of an arbitrary number of measurement settings. Note, moreover, that the case of $s=0,1$ is closest to the original one introduced by Leggett and Garg. We write the probability distribution compactly as $p({\bf q}_{\rm 1\rightarrow n}|{\bf s}_{\rm 1\rightarrow n})$, where ${\bf q}_{\rm 1\rightarrow n}=(q_1,\ldots,q_n)$ and ${\bf s}_{\rm 1\rightarrow n}=(s_1,\ldots,s_n)$.

Consider a variable $Q$, representing a (macroscopic) physical quantity evolving in time. The NIM assumption implies that the marginals of the probability distribution of a whole sequence of measurements are independent of the measurement setting associated with the outcome we sum over. For a sequence of just two measurements and two settings this means $\sum_{q_1}p(q_1,q_2|0,s_2)=p(0,q_2|0,s_2)=\sum_{q_1}p(q_1,q_2|1,s_2)$, where the first equality is just a rephrasing of the fact that for the setting $s_1=0$ there is only one fixed outcome $q_1=0$. 

For a sequence of three measurements MRps and NIM result in the following constraints:
\begin{subequations}\label{eqs:NSIT}
\begin{align}
\nonumber p(0,q_2,q_3|0,s_2,s_3)&=\sum_{q_1}p(q_1,q_2,q_3|1,s_2,s_3),\\
 & \text{ for }(s_2,s_3)\neq (0,0) \label{NSIT1} , \\
p(q_1,0,q_3|s_1,0,1)&=\sum_{q_2}p(q_1,q_2,q_3|s_1,1,1) \label{NSIT2} , 
\end{align}
\end{subequations} 
where again we associate only the outcome $0$ for setting $s_i=0$, i.e., ``no measurement''. As before, the above equalities imply that it is not possible to detect from the observed statistics whether a measurement has been performed and the outcome discarded at some point in the sequence. For example, Eq.~\eqref{NSIT2} demands the impossibility of detecting from the joint statistics at times $t_1$ and $t_3$ whether a measurement has been performed at $t_2$. Consequently, these conditions are called {\it no-signaling in time} (NSIT) conditions \cite{KoflerPRA2013}. 

Further constraints come from the time ordering and are called {\it arrow of time} (AoT) constraints \cite{ClementePRL2016}:
\begin{subequations}\label{eqs:AoT}
\begin{align}
p(q_1,q_2,0|s_1,s_2,0)&=\sum_{q_3}p(q_1,q_2,q_3|s_1,s_2,s_3)\label{AoT1} , \\
p(q_1,0,0|s_1,0,0)&=\sum_{q_2 }p(q_1,q_2,0|s_1,s_2,0)\label{AoT2} , 
\end{align}
\end{subequations} 
which correspond to constraints of no-signalling from the future to the past. These conditions, together with the positivity constraints $p(q_1,q_2,q_3|s_1,s_2,s_3)\geq 0$ define the so-called AoT polytope \cite{ClementePRL2016}. 

The intersection of (AoT) and (NSIT) gives the set of probabilities achievable in macrorealism (MR)~\cite{ClementePRL2016}. In the example above we can already observe that some of the (NSIT) conditions become redundant when we take into account also (AoT). 
\begin{rem}\label{rem:1} Assuming all of Eqs.~(\ref{eqs:AoT}), the set of Eqs.~(\ref{eqs:NSIT}) reduces to the following three independent conditions (in the two settings $s_i=0,1$ case):
\begin{subequations}\label{eqs:NSITT}
\begin{align}
p(0,q_2,q_3|0,1,1)&=\sum_{q_1}p(q_1,q_2,q_3|1,1,1)\label{NSIT11} ,\ \forall q_2,q_3, \\
p(q_1,0,q_3|1,0,1)&=\sum_{q_2}p(q_1,q_2,q_3|1,1,1)\label{NSIT22} ,\ \forall q_1,q_3, \\
p(0,0,q_3|0,0,1)&=\sum_{q_2}p(0,q_2,q_3|0,1,1)\label{NSIT23} , \ \forall q_3,
\end{align}
\end{subequations} 
\end{rem}
{\it Proof.---}Clearly the conditions in Eq.~(\ref{eqs:NSITT}) are a subset of the conditions in Eq.~(\ref{eqs:NSIT}). For the other direction, two cases remain to be checked. First, in Eq.~(\ref{NSIT1}) the case $(s_2,s_3)=(0,1)$ is implied by Eq.~(\ref{NSIT11}), Eq.~(\ref{NSIT22}), and Eq.~(\ref{NSIT23}). Second, in Eq.~(\ref{NSIT1}) the case $(s_2,s_3)=(1,0)$ is implied by Eq.~(\ref{AoT1}) and Eq.~(\ref{NSIT11}).\qed 

The NSIT and AoT conditions can be written for arbitrary sequences. The AoT conditions read
\begin{equation}\label{eq:aot_gen}
p({\bf q}_{\rm 1\rightarrow i}|{\bf s}_{\rm 1\rightarrow i})
=\sum_{q_{i+1}}p({\bf q}_{\rm 1\rightarrow i+1}|{\bf s}_{\rm 1\rightarrow i+1}) ,
\end{equation}
for all $i$, ${\bf q}_{\rm 1\rightarrow i}$ and ${\bf s}_{\rm 1\rightarrow i}$, where we also simplified the notation by not writing the zeros at the end of the sequences of $q$'s and $s$'s. The general NSIT conditions take the form
\begin{equation}\label{eq:nsit_gen}
\begin{split}
&p({\bf q}_{\rm 1\rightarrow i-1},0,{\bf q}_{\rm i+1\rightarrow n}|{\bf s}_{\rm 1\rightarrow i-1},0,{\bf s}_{\rm i+1\rightarrow n})
\\ 
&=\sum_{q_{i}}p({\bf q}_{\rm 1\rightarrow n}|{\bf s}_{\rm 1\rightarrow n}),
\end{split}
\end{equation}
for all $i$, ${\bf q}_{\rm 1\rightarrow n}$ and ${\bf s}_{\rm 1\rightarrow n}$ with $(s_{i+1},\ldots,s_n) \neq (0,\ldots,0)$. Again, the idea is that we cannot detect whether a measurement has been performed, and its outcome discarded, at some point in the measurement sequence.

The notions of NSIT and AoT need not to be restricted to the case of a single physical quantity $Q$ evolving in time. For example, we could have several different measurements at each instant of time, which are not necessarily the time evolved versions of the previous ones. This possibility is still consistent with the notion of macrorealism and noninvasive measurability for multiple physical quantities, as well as with the notion of signalling in time. Moreover, the quantum notion of disturbance is naturally defined in terms of multiple observables rather than only one observable evolving in time. Thus, in the following section we consider the general case, i.e., measurements not necessarily connected by time evolution, and translate the above conditions into a notion of nondisturbance for sequences of quantum measurements. 

\subsection{Quantum nondisturbance}\label{sec:III}

General quantum observables are described by positive-operator valued measures (POVM), i.e., collections of operators $\E=\{E_x\}_x$, such that $E_x\geq 0$ and $\sum_x E_x = \openone$, where $\openone$ denotes the identity operator. Each  $E_x$ is associated to an outcome $x$, and its probability for a state $\rho$ is given the Born rule 
${\rm Prob}(x)=\tr{\rho E_x}$. This gives the outcome probabilities, but not the state transformations after the measurement. In text books of quantum mechanics, state transformations are typically given through the L\"uders' rule~\cite{Luders1950}, i.e., $\rho \mapsto \sqrt{E_x} \rho \sqrt{E_x}$. Such mapping is a special case of a more general set of transformations called quantum instruments. A quantum instrument associated with a POVM $\E$ is a collection of maps $\{ \II_\E^x\}_x $ that satisfies
\begin{align}\label{eq:def_instr}
&\II_\E^x \text{ completely positive (CP)}\nonumber,\\
&\Lambda_{\II_\E}:= \sum_x \II_\E^x \text{ trace preserving (TP)}:\ \Lambda_{\II_\E}^*(\openone)=\openone \nonumber,\\
&(\II_\E^x)^*(\openone) = E_x,
\end{align}
where $(\II_\E^x)^*$ denotes the adjoint of $\II_\E^x$, which maps observables into observables (Heisenberg picture) and $\Lambda_{\II_\E}$ is the total channel corresponding to the sum of 
all instrument elements. Note that the correspondence between POVMs and instruments is one to many. For example, composing each $\II_\E^x$ with a quantum channel $\EE_x$ (CPTP map) 
results in $\EE_x \circ \II_\E^x$, which still satisfies Eq.~\eqref{eq:def_instr}. Interestingly, all instruments associated with the POVM $\E$ are of this form with the initial instrument 
given by the L\"uders instrument, i.e. $\II_{\E_L}^x(\varrho)=\sqrt{E_x}\varrho\sqrt{E_x}$ \cite{PellonpaaJPA2013}.

We can now recall the definition of quantum nondisturbance from Ref.~\cite{HW2010}. A POVM $\A=\{A_x\}_x$ is said to be {\it nondisturbing} with respect to a POVM $\B=\{B_y\}_y$, denoted by $\A \ND \B$, if there exists an instrument $\{\II^x_\A\}_x$ of $\A$ such that\
\begin{equation}\label{eq:def_ND2}
\sum_x \tr{\II_\A^x(\rho) B_y}=\tr{\rho B_y}\ \text{ for all } y \ \text{and } \rho .
\end{equation}
The above condition can be written equivalently in the Heisenberg picture as 
\begin{equation}
\Lambda_{\II_\A}^*(B_y)=B_y\ \text{ for all } y,
\end{equation}
where we have again used the notation $\Lambda_{\II_\A}^*:=\sum_x (\II_\A^x)^*$ for the total channel of $\{\II^x_\A\}_x$. In terms of probabilities, i.e., $p(y|\B)=\tr{\rho B_y}$ and $p(x,y|\A,\B)=\tr{\II_\A^x(\rho) B_y}$, one sees that the above constraint is identical to an NSIT condition, e.g., Eq.~\eqref{NSIT23}, for all possible state preparations $\rho$. If condition \eqref{eq:def_ND2} holds in both directions, i.e., $\A \ND \B$ and $\B \ND A$, we write $\A \NND \B$.

Nondisturbance comes with two special cases that we briefly review together with the basic results on the topic, see Refs.~\cite{HW2010,BuschBook,BuschBook2016}. We call a measurement $\A$ {\it of the first kind}, whenever it is nondisturbing with respect to itself, i.e., $\A \ND \A$. If, furthermore, there exists an instrument $\{\II^x_\A\}_x$ associated to $\A$ such that $(\II_\A^x)^*(A_y)= \delta_{xy} A_y$, where $\delta_{xy}$ is the Kronecker symbol, $\A$ is called {\it repeatable}. Repeatable measurements are characterised as POVMs with every POVM element having an eigenvalue 1. Separate sufficient conditions for the first kind property are commutativity and repeatability. Commutativity is also necessary in the qubit case and in the case of rank-1 measurements.

In the case of two nonequal POVMs commutativity is again sufficient for nondisturbance (and also necessary in the qubit case and in the case of rank-1 POVMs) \cite{HW2010}. However, in this case there are also other notions of compatibility that relate to nondisturbance. We recall here the notion of joint-measurability, together with the formal notion of commutativity, for reference in the following discussion:

(i) Two POVMs $\A=\{A_x\}_x$ and $\B=\{B_y\}_y$ are {\it commuting} whenever $[A_x, B_y]=0$ for all pairs of outcomes $(x,y)$. \\
(ii) Two POVMs $\A=\{A_x\}_x$ and $\B=\{B_y\}_y$ are {\it jointly measurable} whenever there exists a POVM $\Go=\{G_{xy}\}_{xy}$ 
such that $A_x=\sum_y G_{xy}$ for all $x$ and $B_y=\sum_x G_{xy}$ for all $y$.

Note that when at least one of the two POVMs is projective the above definitions of incompatibility are equivalent to each other and to nondisturbance. However, in the general case the concepts form a strict hierarchy: commutativity implies  nondisturbance, which, in turn, implies joint-measurability. To be more concrete, for commuting POVMs the L\"uders rule gives a nondisturbing implementation. Moreover, nondisturbance between $\A$ and $\B$ (in one direction) implies joint-measurability as any nondisturbing instrument implements a POVM, i.e. $\{(\II_\A^x)^*(B_y) \}_{xy} = \{G_{xy}\}_{xy}$, satisfying the requirements of a joint observable. The inverse implications are not true in general \cite{HW2010}.

In contrast to the other two notions, for nondisturbance the instruments are taken into account. As a consequence, deciding nondisturbance becomes somewhat more complicated, especially for longer sequences. In fact, efficient methods are available for the case of two measurements~\cite{HW2010} but only heuristic methods for longer sequences, see the discussion in Appendix~\ref{app:semidef}.

\section{Main results: Macrorealism in terms of quantum disturbance}\label{sec:IV}

\subsection{State-independent NSIT conditions from quantum nondisturbance}

Here, we present the formulation of a state-independent and clumsiness-free version of the NSIT conditions, given in terms of a generalization of quantum nondisturbance conditions.
For NSIT conditions, the simplest scenario is given by a sequence of two measurements $\Q_1$ and $\Q_2$, and two inputs $s_i=0$ and $1$ corresponding, respectively, to ``no measurement'' and ``measurement of $\Q_i$''. In this case we have only one NSIT condition, namely, $p(0,q_2|0,s_2)=\sum_{q_1}p(q_1,q_2|s_1, s_2)$. The above condition can be expressed in quantum theory by means of the Born rule as
\begin{equation}
\tr{\rho Q_2^{q_2}} =  \tr{\Lambda_{\II_{\Q_1}}(\rho) Q_2^{q_2}} \ \forall q_2,
\end{equation}
where  $\Lambda_{\II_{\Q_1}}=\sum_{q_1}\II_{\Q_1}^{q_1}$ is the total channel.
When optimizing over all states and instruments, the quantum expression of the NSIT condition becomes equivalent to the  nondisturbance condition $\Q_1 \ND \Q_2$.

For the case of three measurements in a sequence, the conditions on macrorealism can be mapped into three different NSIT conditions, i.e., Eqs.~(\ref{eqs:NSITT}), for probabilities $p(q_1,q_2,q_3|s_1,s_2,s_3)$. Since  we are interested in properties of observables, we optimize over all states. Moreover, in order to separate the clumsiness of a specific measurement implementation from the fundamental limitations on measurement disturbance predicted by quantum theory, we also perform an optimization over measurement implementations. The following observation connects the NSIT conditions with minimal nondisturbance requirements on quantum instruments (see Fig.~\ref{fig:Fig2}):

\begin{figure}[t]
\includegraphics[width=0.35\textwidth]{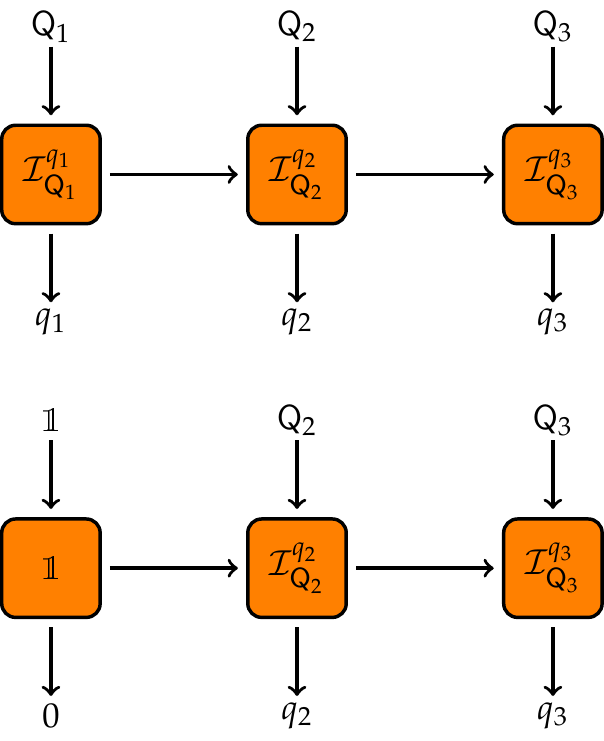}
\caption{A sequence of measurements which generalizes the example in Fig.~\ref{fig:Fig1}: General POVMs $(\Q_1,\Q_2,\Q_3)$ are considered instead of projective measurements; general quantum instruments $\II_{\Q_i}^{q_j}$ are considered instead of the L\"uders' rule; the two observables are not necessarily connected by time evolution. When optimization over initial preparations is performed, these two sequences form the minimal set of macrorealism conditions.}\label{fig:Fig2}
\end{figure}

\begin{obs}\label{obs:D_NS2} Consider a sequence of three POVMs labelled by $\Q_1, \Q_2$ and $\Q_3$. They allow for an implementation satisfying the NSIT conditions for all states if and only if there exists a pair of instruments $\II_{\Q_1}$ and $\II_{\Q_2}$ such that 
\begin{subequations}\label{eq:minimalNSIT}
\begin{align}
\tr{\II_{\Q_2}^y(\rho)Q_3^z}&=\sum_{x}\tr{\II_{\Q_2}^y\circ\II_{\Q_1}^x(\rho) Q_3^z} , \label{NSIT1a} \\
\tr{\rho Q_3^z}&=\sum_{y}\tr{\II_{\Q_2}^y(\rho)Q_3^z} , \label{NSIT4a}
\end{align}
\end{subequations}
for all $\rho$ and all $(y,z)$. More compactly, we write
\begin{equation}
\begin{split}
\left\lbrace
\begin{array}{l}
\Q_1 \ND \II^*_{\Q_2}\Q_3 , \\
\Q_2 \ND \Q_3 ,
\end{array}
\right.
\end{split}
\end{equation}
for some instrument $\II_{\Q_2}$ implementing $\Q_2$, which is the same in both sequences.
\end{obs}
{\it Proof.---} Eqs.~(\ref{eq:minimalNSIT}) correspond to Eqs.~(\ref{NSIT11},\ref{NSIT23}) for a fixed state $\rho$ and measurements $\{\Q_i\}$ with instruments $\{\mathcal I_{\Q_i}^{x}\}$. It is sufficient to show that Eq.~(\ref{NSIT22}) becomes redundant when we require the other conditions to hold for all state preparations.
In fact, let us consider Eq.~(\ref{NSIT4a}) with a particular (unnormalized) state~\footnote{Note that the different normalization of the state does not change anything in the conditions above, since it gives the same factor for both sides of the equation.} $\rho=\II_{\Q_1}^x(\sigma)$ obtained from a certain $\sigma$ and a certain instrument $\II_{\Q_1}^{x}$. Then, we have that 
\begin{equation}
\begin{aligned}\label{eq:proofobs1}
\sum_y p(x,y,z|q_1,q_2,q_3)&=\sum_y \tr{\II_{\Q_2}^y\circ \II_{\Q_1}^x(\sigma)Q_3^z} \\
=\tr{\II_{\Q_1}^{x}(\sigma) Q_3^z}&=p(x,0,z|Q_1, 0, Q_3), 
\end{aligned}
\end{equation}
which corresponds precisely to Eq.~(\ref{NSIT22}) for a state $\sigma$ and  instrument $\II_{\Q_1}^{x}$. 
Thus, we have that if Eq.~(\ref{NSIT4a}) holds for all states $\rho$ and some instrument $\II_{\Q_2}^{y}$, then also  Eq.~(\ref{eq:proofobs1}) holds for the same $\II_{\Q_2}^{y}$ and all $\II_{\Q_1}^{x}(\sigma)$.
\qed

\begin{rem}\label{rem:2} 
For a sequence of POVMs $(\Q_1, \Q_2, \Q_3)$ with a fixed order, Eq.~\eqref{NSIT4a} and  Eq.~\eqref{NSIT1a} are independent conditions.
\end{rem}
{\it Proof.---} Choose a triple of POVMs such that $\Q_1$ disturbs $\Q_2$, e.g. noncommuting projective measurements, and take $\Q_3$ to be the coin flip POVM, i.e., $\Q_3=\{\frac{1}{2}\openone,\frac{1}{2}\openone\}$. Now $\Q_2 \ND \Q_3$ while $\Q_1$ disturbs the sequence $\II_{\Q_2}^* \Q_3=\frac{1}{2}\Q_2$. Similarly, one can choose  $\Q_1$ as the coin flip POVM, which always have a nondisturbing implementation, and $\Q_2$ that disturbs $\Q_3$.\qed

Note that for triples of POVMs with no fixed ordering the conditions remain minimal, see Obs.~\ref{obs:hollow_triang} below.  Obs.~\ref{obs:D_NS2} can be straightforwardly generalized to arbitrary sequences.
\begin{obs}\label{obs:2}
Let us consider a sequence of $n$ measurements described by POVMs $\Q_1,\ldots, \Q_n$. These POVMs allow an implementation satisfying the NSIT conditions in Eq.~\eqref{eq:nsit_gen} for any preparation if and only if:
\begin{equation}\label{eq:fullMRcondnpoint}
\begin{split}
\left\lbrace
\begin{array}{l}
\Q_1 \ND \II_{\Q_2}^*\II_{\Q_3}^*\ldots \Q_n\\
\Q_2 \ND  \II_{\Q_3}^*\II_{\Q_4}^*\ldots \Q_n\\
\ \ \vdots \qquad \vdots \\
\Q_{n-1} \ND \Q_n.
\end{array}
\right.
\end{split}
\end{equation}
\end{obs}
A detailed proof is presented in Appendix~\ref{App:B}. Note that the minimality argument in Remark~\ref{rem:2} (for a fixed ordering), constructed with a pair of disturbing observables and a trivial one, can be easily generalized to the case of $n$ measurements in Obs.~\ref{obs:2}, by using $n-2$ trivial observables and two disturbing ones. In fact, trivial observables never disturb (with a proper choice of instrument), and are never disturbed by, other observables. As a consequence, one obtains that Eq.~\eqref{eq:fullMRcondnpoint} also provides a minimal set of conditions.

Finally, we can extend the above conditions to the case in which all orders are possible:
\begin{obs}\label{obs:3}
Let $\II_{\Q_1},\ldots, \II_{\Q_n}$ be measurement implementations of the POVMs $\Q_1.\ldots,\Q_n$. The NSIT conditions are satisfied for all permutations of the sequence of instruments and for all state preparations, if and only if
\begin{equation}\label{NDorder}
\left\lbrace
\Q_{\pi(1)} \ND \II_{\Q_{\pi(2)}}^*\II_{\Q_{\pi(3)}}^*\ldots \Q_{\pi(n)}
\right\rbrace_{\pi}
\end{equation}
holds for all permutations $(\pi(1),\ldots,\pi(n))$. Note that we have slightly abused the notation here: by $A\ND B$ we mean that the fixed instrument is an optimal one, i.e. no optimisation over instruments is performed.
\end{obs}
A detailed proof is presented in Appendix~\ref{App:B}. As a final remark, notice that not all pairs of POVMs can be connected by time evolution. In fact, our nondisturbance conditions differ for the case of time-evolved observables, where macrorealism translates to slightly weaker constraints (see Appendix~\ref{App:A}).

\subsection{Quantum nondisturbance and the clumsiness loophole}\label{sec:clumsloop}

The question arises of how the notion of NSIT and quantum nondisturbance introduced above relates to other notions of 
weakly, or minimally, disturbing instruments and their role in tests of macrorealism.  Several approaches have been proposed such as {\it weak measurement} \cite{Dressel_rev_2014, WV_Review}, {\it ideal negative-choice measurement}, introduced already by Leggett and Garg~\cite{LeggettPRL1985}, {\it adroit measurements}, introduced by Wilde and Mizel~\cite{WildeMizel2012}, the notion of {\it nondisturbing measurement} by George {\it et al.}~\cite{GeorgePNAS2013},or some other forms of quantification of the disturbance with extra {\it control experiments} \cite{KneeNC2016,EmaryPRA2017}.
In the following, we compare our approach based on NSIT with different notions of ``minimally invasive'' measurements and show how our framework can extend previous approaches and stimulate the design of stronger experimental tests of nonmacrorealism.

Let us start with weak measurements and, more precisely, with a concrete model for them, that we call $\W(\Q)$. This is defined by means of an ancillary system, usually considered as continuous variable with ``position'' $q$, in a Gaussian state $|\phi(q)\rangle=\int{\rm d}q \  \phi(q) |q\rangle $ with $\phi(q)=(2\pi s)^{-1/4}\exp(-q^2/4s^2)$ being a Gaussian wave function with standard deviation $s$. Then, the canonical form of a POVM associated to the weak measurement of a (PVM) observable $\Q$ is given by \cite{WV_Review} (see also the discussion in \cite{Bednorz2012, EmaryRPP2014} for LG inequalities) $\W_q(\Q)= K_q^\dagger K_q$, with Kraus operator
\begin{equation}\label{eq:Kq_def}
K_q = (2\pi \sigma)^{-1/4}\exp(-(q-\Q)^2/4s^2) ,
\end{equation}
where $\Q=\sum_x x\ketbra{x}$ is an observable given by a Hermitian operator. Here, $s$ represents the ``weakness'' of the measurement, and in the limit $s \rightarrow \infty$ one obtains a weak measurement, while in the limit $s \rightarrow 0$ the measurement is ``strong'', i.e., it becomes the PVM $\Q$. Note also that the outcomes $q$ are referred to the ancillary system and are different from $x$. In fact, $q$ is usually a continuous variable whereas $x$ can be from a discrete set of outcomes.

The weak measurement scheme implements a POVM with elements that are functions of the Hermitian operator $\Q$. Therefore, measurement that do not disturb (or are not disturbed by) $\Q$, will do the same with $\W(\Q)$. 
More precisely, the POVM $\W(\Q)$ can be seen as just a classical postprocessing of $\Q$, i.e., $\W_q=\sum_x p(q|x) \ketbra x$, where $p(q|x)=(2\pi \sigma)^{-1/2}\exp(-(q-x)^2/2s^2)$ is a classical Gaussian probability distribution. In this sense, weak measurements such as $\W(\Q)$ do not seem to explore the rich structures of incompatibility relations among POVMs.
In this respect, see also the discussion in Sec.~\ref{sec:resourceperspective} about classical postprocessing as a free operation in a possible resource theory approach to violations of macrorealism.

Moreover, as discussed by Emary {\it et al.}~\cite{EmaryRPP2014} the strength or weakness of the measurement is not necessarily related to the noninvasivity, but rather to the ambiguity of the result. 
As a final remark, note that even when measurements are performed through a weak interaction between the system and the ancilla, such weakness is only justified from a quantum mechanical perspective and has, in principle, no meaning for a macrorealist, who would not accept a quantum mechanical description and require a definition of strength in terms of some explicitly experimental procedure \cite{EmaryRPP2014,WildeMizel2012}. 

In contrast, the notion of ideal negative-choice measurement (INCM) is fully defined in terms of macrorealist theories. Let us consider as an example the INCM by Knee et al.~\cite{KneeNC2012}. The authors performed an INCM by coupling the system to an ancilla via a (anti-)CNOT gate and measuring the ancilla. The main assumption of this INCM is that when the system is in the state $\ket{0}$ (or $\ket{1}$ for the case of anti-CNOT), the measurement procedure will not interact with the system. It is true that the state is unchanged for that particular preparation ($\ket{0}$ for the CNOT and $\ket{1}$ for the anti-CNOT), however, the measurement is not only disturbing from a quantum mechanical perspective, but it gives rise to explicit signaling in time for other preparations of the initial state,  a fact that a macrorealist could hardly ignore.

The proposal of Wilde and Mizel~\cite{WildeMizel2012} consists of the notion of {\it adroit measurements}, namely, measurements that are unable to violate a LGI when taken alone, but only when combined together. More precisely, in the usual scenario with three measurements $\Q_1,\Q_2,\Q_3$, the intermediate measurement is decomposed in four operations denoted as $\OO_1, \OO_2, \OO_3, \OO_4$. The decomposition is such that for any measurement sequence $\Q_1, \OO_i, \Q_3$ Eq.~\eqref{NSIT22} is satisfied, i.e., $\OO_i$ does not disturb the probability $p(q_1,q_3)$ or equivalently one can never violate the corresponding three-term LGI. The idea is to choose projective measurements with $Q_j,O_i$ denoting the corresponding Hermitian operators, such that $[O_i, Q_1]=0$ for odd $i$ and $[O_i,Q_3]=0$ for even $i$, and $\Q_2$ is given by the sequential measurement of $\OO_1,\OO_2,\OO_3,\OO_4$ with a summation over the first three outcomes. The conclusion of Wilde and Mizel is that each single measurement seems to be nondisturbing, in the sense of unable to violate a LGI, nevertheless their combination violate a LGI. The NIM assumption is then relaxed to the {\it noncolluding measurements} (NCM) assumption: seemingly nondisturbing measurements must necessarily ``collude'' to violate a LGI. However, even if each triple $\Q_1,\OO_i,\Q_3$ is unable to violate a LGI, it is very easy to show that $\OO_1$ and $\OO_3$ disturb $Q_3$ and $\Q_1$ disturbs $\OO_2,\OO_4$, i.e., they violate Eq.~\eqref{NSIT23}, since they are noncommuting projective measurements. 

In addition to the ideal case, Wilde and Mizel also extended their discussion to the case of $\varepsilon$-adroit measurements, where measurements show a $\varepsilon$-disturbance according to their measure, i.e., the single LGI tests of $\Q_1, \OO_i, \Q_3$ shows a violation $\varepsilon$. This quantification of the disturbance via extra {\it control experiments}, can then be used, under appropriate assumptions, to modify the classical bound for LG expressions. The proposal has been further extended by other authors. For example, in Knee et al.~\cite{KneeNC2016} the invasivity of the measurements appearing in the LG test is quantified through extra control experiments. There, in order to violate macrorealism, one no longer needs to assume noninvasive measurability, but rather a weaker condition, namely that the description of the systems utilizes only two classical states.

Many of the experimental proposals and tests of LGI and NSIT conditions involve qubit system, where the possible incompatibility structures are limited~\cite{HW2010}. George {\it et al.}~\cite{GeorgePNAS2013} proposed a scheme that uses qutrit observables. More precisely, they design a 3-term LGI based on the three-box paradox~\cite{Aharonov_1991}, for which the intermediate measurements seem to be ``nondisturbing'' w.r.t. the third for the state preparation used for the experiment. However, the observables they measure are noncommuting PVMs and thus there exist state preparations for which they violate the NSIT \eqref{eqs:NSITT}.

Our approach to NSIT and quantum nondisturbance allows us to generalize and treat in a systematic way this type of arguments. We wish to prove the existence of the following type situation: three measurement $\Q_1,\Q_2,\Q_3$ are given with fixed implementations. The implementations are such that any sequence of two measurements regardless of the measurement order is non-disturbing for all input states, i.e., all NSIT conditions are satisfied. The sequence $\Q_1\rightarrow \Q_2\rightarrow \Q_3$ is such that the NSIT conditions are violated for some input state. By providing measurements satisfying the above conditions, one can design experiments that further tighten the clumsiness loophole and consequently disprove a larger set of macrorealist theories. We are going to show in the following, Obs.~\ref{obs:hollow_triang}, that such a construction is possible.

In analogy with the case of joint measurability \cite{JM_struct}, we further conjecture that more complex structures are possible, namely, measurements that are nondisturbing for any possible subset of $n-1$ measurements, in any order and for any state, but violate an $n$-term NSIT when measured all together. The design and implementation of such measurements would allow one to substitute the assumption of noninvasive measurement with a hierarchy of weaker assumptions defined as follows. In analogy with Wilde and Mizel's notion, we define {\it $n$-term adroit measurements} as a set of measurements that satisfy the NSIT conditions for any initial state and for all possible sequence and ordering of at most $n$ measurements, then substitute the NIM assumption with the assumption
\begin{itemize}
\item ($n$-NCM) {\it $n$-term noncolluding measurements}:\\
A set of measurements that are  $(n-1)$-term adroit are also $n$-term adroit.
\end{itemize}

In summary, our systematic approach to the relation between quantum nondisturbance and macrorealism allows us to highlight the weaknesses of previous approaches, such as, e.g., the focus on particular state preparations~\cite{KneeNC2012,GeorgePNAS2013} or necessary but not sufficient conditions for macrorealism~\cite{WildeMizel2012},  and, at the same time, extend such approaches to further tighten the clumsiness loophole, by weakening the assumption of noninvasive measurements and allowing for the development of a hierarchy of notions of noncolluding measurements.

\subsection{Structure of nondisturbance relations}\label{sec:V}

Here we explore the richness of multiplewise quantum nondisturbance, and find examples of sequences of POVMs that reveal the existence of structures that can be also exploited to tighten the clumsiness loophole, as we explained above. The notion of nondisturbance 
can be considered in between commutativity and joint measurability. Commutativity is a binary relation of compatibility with the peculiary property, also known as Specker's principle \cite{Specker1960,Cabello_sp2012}, that mutually compatible pairs are also ``globally'' compatible. 
In contrast, JM allows all possible structures of (in)compatibility \cite{JM_struct}.

Hence, the question arises whether nondisturbance is a binary relation or not, at least whenever the notion is ``symmetrized'' between two observables, i.e, whether observables that are mutually nondisturbing $\A \NND B\NND C \dots$ are also nondisturbing in every possible (full) sequence $\A \ND \B \ND C \dots$ etc.
To answer this, at first it is helpful to show that mutual nondisturbance and commutativity are stricly distinct notions.
\begin{obs}\label{obs:2wNDcomm}
There exists observables $\A$ and $\A'$ such that they are nondisturbing in both directions, i.e., $\A \NND \A^\prime$ and yet non commuting.
\end{obs} 
The proof consists in a specific example and can be found in Appendix~\ref{App:B}.

A natural question, then, arises of whether mutual nondisturbance can give rise to nontrivial structures for longer sequences of measurements, such as the so-called ``hollow triangle'' \cite{Specker1960, KunjwalPRA2014,Yu2013holl} (see Fig.~\ref{fig:Fig3}).
\begin{obs}\label{obs:hollow_triang}
There exists observables $\A, \B,\C$ such that $\A \NND \B$, $\A \NND \C$, and $\C \NND \B$, but in the sequence $\A \rightarrow (\B \rightarrow \C)$, $\A$ disturbs any nondisturbing implementation of the sequence $\B \rightarrow \C$.
\end{obs}
The proof is rather technical and we omit it here, details can be found in Appendix~\ref{App:B}. 

As we anticipated in the previous subsection, the construction in Obs.~\ref{obs:hollow_triang} allows us to define a new notion of $2$-term adroit measurements. In this case, there exists an implementation of the observables (i.e., a quantum instrument), such that for each pair, in each order, and for each state preparation, the measurements are nondisturbing, but still able to violate a NSIT condition for the sequence  $\A \rightarrow (\B \rightarrow \C)$. This setup would allow for an experimental test of macrorealism where the NIM assumption is substituted with the weaker $3$-NCM discussed previously.
Moreover, our approach to NSIT and quantum nondisturbance allows us to treat the problem from a general perspective and systematically investigate possible structures of incompatibility, which may give rise to stronger experimental tests, i.e., based on weaker assumptions, of macrorealism.

\begin{figure}[t]
\includegraphics[width=0.5\textwidth]{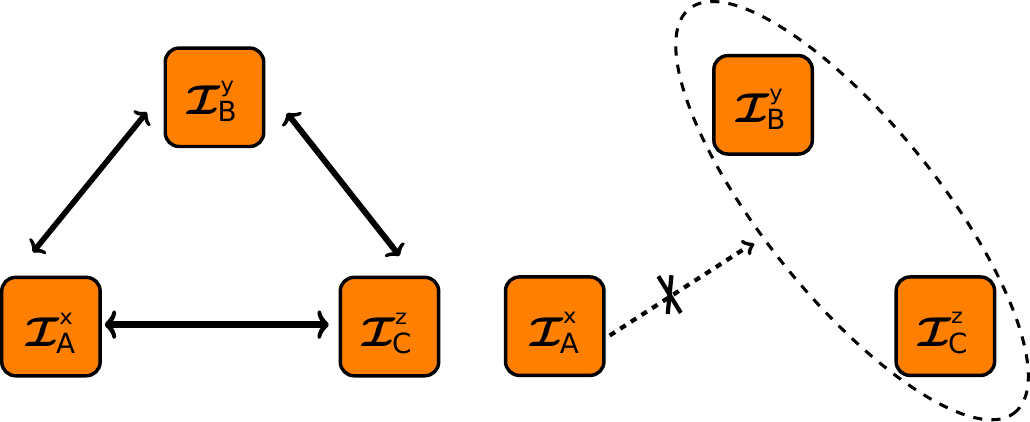}
\caption{Diagram of the hollow triangle of nondisturbance as in Obs.~\ref{obs:hollow_triang} even though there are instruments that allow nondisturbing implementations of every pairwise sequence, there is no triple of instruments $(\II_{\A}, \II_{\B}, \II_{\C})$ that allows a nondisturbing implementation for the sequence $\A \rightarrow (\B \rightarrow \C)$ or the sequence $\A \rightarrow (\C \rightarrow \B)$.}\label{fig:Fig3}
\end{figure}

\section{Perspectives on incompatibility quantification as intrinsic disturbance}\label{sec:resourceperspective}

It is known that notions of measurement incompatibility are related to different nonclassical properties of quantum systems such as nonlocality \cite{WolfPRL2009} and steering \cite{UoMoGu14, QuVeBr14, UoBu15} and consequently to the associated quantum information applications, such as, e.g., quantum key distribution~\cite{AcinQKD2007, BranciardQKD2012}. Different notions of measurement incompatibility are expected to play a role in the generation of nonclassical correlations for single systems. A notion of nonclassical correlations for single systems, quantum contextuality~\cite{Kochen1967,KlyachkoPRL2008,CabelloPRL2008,CSW2014}, has been suggested as a resource for quantum computation~\cite{Howard2014, RaussPRA2013,DelfossePRX2015,BermejoPRL2017,RaussPRA2017} and the a resource theory of quantum contextuality has been proposed~\cite{GrudkaPRL2014,AbramskyPRL2017,AmaralPRL2018}. Moreover, the connection between some notions of temporal correlations and information-theoretic tasks has been investigated via several different approaches~\cite{GallegoPRL2010,BrunnerPRL2013,GuehnePRA2014,BudroniPRL2014,DallarnoPRSA2017,DallarnoPRL2017,RossetPRX2018,BowlesPRA2015,TavakoliPRA2016,
AguilarPRL2018, Miklin2019, Hoffman2018, Spee2018}, but not in relation with the notion of incompatibility. 
Macrorealism violations have also been investigated in relation to resource theories of quantum coherence~\cite{Knee2018,Moreira2019}.

Following these approaches, we investigate a possible resource theory of macrorealism and its connection with the natural notion of measurement incompatibility in the temporal correlation scenario, namely, quantum disturbance.
We propose a possible quantification of intrinsic disturbance, based on the conditions in Eq.~\eqref{eq:fullMRcondnpoint}, and study the operations that do not increase it. We call them ``free processes'' or ``free operations'' following the terminology of resource theories (see, e.g., Ref.~\cite{ResTh2016}). Our measure can be seen as an extension to arbitrary sequences of the measure introduced in Ref.~\cite{HW2010} for pairs of measurements~\footnote{Notice, however, that in contrast to Ref.~\cite{HW2010} we sum over the disturbances on single POVM elements rather than picking the maximally disturbed element}.

A natural question arises of whether the notion of disturbance may be of interest from a quantum information perspective. In this sense, we believe that the term ``quantum disturbance'', even though is the standard term in the literature, is misleading as it indicates some sort of ``noise'' or ``destruction'' of the quantum information. In contrast, the ability to ``disturb'' subsequent measurements must be interpreted as the possibility of encoding information into the system, which is then decodable by subsequent measurements. Consider the example of a quantum information protocol, e.g., computation, performed via a sequence of measurements. Intuitively, the case in which only macrorealist correlations arise can be efficiently simulable since such correlations are ``fixed'' at the beginning of the protocol and do not depend on the act of measuring some particular observable.  Any classical simulation of this situation would involve additional memory that must be updated after each measurement to keep track of the (detectable) changes to the system, in contrast to usual hidden variable model description where the outcomes of all possible measurements are decided at the beginning of each experimental run.

Given the hierarchical structure of nondisturbance conditions, let us start from the simplest scenario. For two POVMs $\A$ and $\B$ and fixed instrument $\II_\A$ we define
\begin{equation}\label{eq:DinstTwo}
D_{\A}(\B, \II_\A):=  \left[ \sum_y \Big\| \Lambda_{\II_\A}^*(B_y)- B_y \Big\| \right] ,
\end{equation}
where $\| X \| = \sup_\rho |\tr{\rho X}|$ is a matrix norm, i.e., the maximal eigenvalue (in absolute value) of a Hermitian operator $X$, and $\Lambda_{\II_\A}^*:=(\sum_x\II_\A^x)^*$. The corresponding \textit{disturbance measure} can be written as the infimum over all instruments implementing $\A$ as
\begin{equation}\label{eq:DAB_def}
D_{\A}(\B):= \inf_{\II_{\A}} D_{\A}(\B, \II_\A) .
\end{equation} 

For the case of finite dimensional Hilbert spaces, and analogously to Ref.~\cite{HW2010}, such a measure can be cast as a semidefinite program (see Appendix~\ref{App:C}). The same expression has been used as a macrorealist quantifier for projective instruments in Ref.~\cite{SchildPRA2015}.
In a sequence of two measurements $\Q_1$ and $\Q_2$ with no fixed order, one can define a symmetrized version of the measure, i.e., we define a \textit{measure of violation of macrorealism} as
\begin{equation}\label{eq:def_MR3}
\begin{aligned}
{\rm MR}(\Q_1,\Q_2)&:= D_{\Q_1}(\Q_2) + D_{\Q_2}(\Q_1) .
\end{aligned}
\end{equation}

Instead of explicitly writing the measure of disturbance and the measure of violation of macrorealism for arbitrary sequences, we focus only on the case of three measurements, as the idea of how to generalize the measure to longer sequences becomes evident already from this case. For this scenario, macrorealism is associated to the independent (see Remark \ref{rem:2} above) conditions in Eq.~(\ref{NSIT1a}) and Eq.~\eqref{NSIT4a}.
Hence, we write the measure of violation of macrorealism as
\begin{equation}\label{Macromeas3}
\begin{split}
{\rm MR}(\{\Q_i\}_{i=1}^3):=\sum_{{\rm perm.}~\pi } \inf_{\{\II_{\Q_i}\}_i}  
 \left[  D_{\Q_{\pi(2)}}(\Q_{\pi(3)}, \II_{\Q_{\pi(2)}})   \right. \\
\left. + D_{\Q_{\pi(1)}}(\Q_{\pi(2)}, \Q_{\pi(3)}, \II_{\Q_{\pi(1)}}, \II_{\Q_{\pi(2)}} ) \right],
\end{split}
\end{equation}
where the infimum is taken over the instruments implementing the corresponding POVMs, the sum is over all possible ordering of measurements, i.e., permutations of $(1,2,3)$, and we have denoted
\begin{equation}
D_{\A}(\B, \C, \II_{\A},\II_{\B}):=  \sum_{yz} \Big\|\Lambda_{\II_\A}^*((\II_\B^*\C)_{yz})- (\II_\B^*\C)_{yz} \Big\| ,
\end{equation}
where $\Lambda_{\II_\A}^*:=(\sum_x\II_\A^x)^*$ and $(\II_\B^*\C)_{yz}:=(\II_\B^y)^*(C^z)$.

The definition in Eq.~\eqref{Macromeas3} can be straightforwardly generalized to arbitrary sequences, by using conditions in Eq.~\eqref{eq:fullMRcondnpoint}.

In order to build a resource theory associated to the intrinsic disturbance or the violation of macrorealism one needs to define the free resources and free operations. Free resources are naturally identified with sequences of POVMs that are nondisturbing according to  the measure in Eq.~\eqref{Macromeas3} (and its generalizations to arbitrary sequences), i.e., they always admit a physical implementation that does not provide violation of macrorealism for any possible state preparation. Free operations, instead, are operations on the POVMs that do not increase the measure ${\rm MR}(\Q_1,\dots,\Q_n)$. We have
\begin{obs}\label{obs:free_op}
The following transformations of POVMs are free operations:
\begin{itemize}
\item[$(i)$] global classical postprocessing: $\A \mapsto \tilde{\A}$ with $\tilde A_{a}=\sum_{a'}p(a|a')A_{a'}$ and $p(\cdot|a')$ a probability distribution for every $a'$, namely, ${\rm MR}(\{\td \Q_i\}_{i=1}^n)\leq {\rm MR}(\{\Q_i\}_{i=1}^n)$;
\item[$(ii)$] global invertible quantum preprocessing: $\A \mapsto \tilde{\A}$ with $\tilde A_{a} = \UU(A_{a})$, with $\UU$ unitary channel, namely ${\rm MR}(\{\td \Q_i\}_{i=1}^n)\leq {\rm MR}(\{\Q_i\}_{i=1}^n)$;
\item[$(iii)$] global quantum preprocessing for qubit channels: $\A \mapsto \tilde{\A}$ with $\tilde A_{a} = \Lambda(A_{a})$, with $\Lambda$  a qubit channel, namely ${\rm MR}(\{\td \Q_i\}_{i=1}^n)\leq {\rm MR}(\{\Q_i\}_{i=1}^n)$;;
\item[$(iv)$] local preprocessing via a depolarazing channel: $\A\mapsto \td \A$ with $\tilde{A}_a = \alpha A_a + \tfrac{1-\alpha} d \tr{A_a} \openone$ for $\alpha\in[0,1]$, namely ${{\rm MR}( \Q_1,\ldots, \td \Q_i,\ldots, , \Q_n) \leq {\rm MR}(\Q_1,\ldots, \Q_i,\ldots  \Q_n)}$.
\end{itemize} 
\end{obs}
See Appendix~\ref{App:C} for the proof. Notice that classical postprocessing is precisely the operation defining the weak measurement in Eq.~\eqref{eq:Kq_def}. It is not clear whether global quantum preprocessing is, in general, a free operation. For the specific case of the depolarizing channel, we were able to prove that even its local application is a free operation.

\section{Conclusion and outlook}\label{sec:conclusion}
In summary, we investigated a systematic approach to the notion of quantum nondisturbance and its relation to the NSIT conditions of macrorealism. We first derived a minimal set of conditions able to guarantee NSIT for any state preparation. Then, we discussed how such systematic approach to quantum nondisturbance conditions, together with structural results we derived, could be used to tighten the clumsiness loophole and go beyond the limitations of previous approaches. Finally, we discussed a possible resource theory of quantum nondisturbance and macrorealism, by defining a measure of disturbance and investigating which operations do not increase it.

A natural question arises of why one should care about the relation between incompatibility properties and notions of nonclassical temporal correlations such as nonmacrorealism. First, let us recall that 
in the spatial scenario, quantum incompatibility among observables has been proven to be at the basis of quantum phenomena of nonlocality \cite{WolfPRL2009} and steering \cite{UoMoGu14,QuVeBr14,UoBu15}, which, in turn, are at the basis of quantum information theoretic tasks such as quantum key distribution \cite{AcinQKD2007,BranciardQKD2012}. The temporal analogue, however, is less explored, from the perspective of notions of incompatibility, nonclassical correlations, and quantum information theoretic tasks. A notion of nonclassical correlations for single systems, namely quantum contextuality \cite{Kochen1967,KlyachkoPRL2008,CabelloPRL2008,CSW2014}, has been recently suggested as a resource for quantum computation \cite{Howard2014, RaussPRA2013,DelfossePRX2015,BermejoPRL2017,RaussPRA2017} and investigated from the perspective of resource theories \cite{GrudkaPRL2014,AbramskyPRL2017,AmaralPRL2018} and the perspective of memory cost for its simulation \cite{KleinmannNJP2011,Fagundes2017}. Moreover, the connection between temporal correlations and information-theoretic tasks has been investigated via several different approaches
\cite{GallegoPRL2010,BrunnerPRL2013,GuehnePRA2014,BudroniPRL2014,DallarnoPRSA2017,DallarnoPRL2017,
RossetPRX2018,BowlesPRA2015,TavakoliPRA2016, AguilarPRL2018, Hoffman2018, Spee2018}, but never in relation with the notion of incompatibility. 

Starting from our results, several other lines of research can naturally be pursued. For instance, we showed that all but one of the conditions in Obs.~\ref{obs:2} can (trivially)  be satisfied, impliying that one can find sequences of measurements such that all sequences but one are nondisturbing, when the order of measurements is fixed. With the construction in Obs.~\ref{obs:hollow_triang} we showed that this result can be extended to arbitrary orders for $n=3$. Can this result be extended to arbitrarily long sequences? Do the most general incompatibility structures exist, as it is the case for joint measurability \cite{JM_struct}, or are there some limitations?  Finally, can stronger tests of macrorealism be designed via such constructions and the $n$-NCM assumption? 

In parallel, we would like to further investigate the connection between disturbance and quantum information tasks. Can our resource theory be connected to other resource theories of nonclassical correlations? How is the notion of disturbance related to the notion of memory cost of simulating sequential measurements? 

We hope our work will stimulate the investigation of the connection between macrorealism, quantum incompatibility, and quantum information tasks in the temporal scenario, and create a common basis for discussions among the corresponding research communities.

\acknowledgements
The authors thank Marco T\'ulio Quintino, Teiko Heinosaari, and Tom Bullock for stimulating discussions. This work has been supported by the the Austrian Science Fund (FWF): M 2107-N27 (Lise-Meitner Programm), M 2462-N27 (Lise-Meitner Programm),  START project
Y879-N27,  ZK 3 (Zukunftskolleg),  the Finnish Cultural Foundation, and the ERC (Consolidator Grant No. 683107/TempoQ).

\appendix

\section{Technical proofs of Obs.3-5}\label{App:B}
{\it Proof of Obs.\ref{obs:2}.---} One direction is trivial as the NSIT constraints (holding for all states) include all conditions in Eq.~\eqref{eq:fullMRcondnpoint}.

The opposite direction, i.e., that the conditions in Eq.~\eqref{eq:fullMRcondnpoint} are enough to imply all NSIT conditions for all state preparations, can be obtained by induction on the length of the sequence, following the idea of the proof of Obs.~\ref{obs:D_NS2}. By Obs.~\ref{obs:D_NS2}, we know that the statement holds for the case $n=3$. The inductive step will be to prove that if the statement holds for sequences of length $n-1$, then it holds for sequences of length $n$.

Thus, let us consider a sequence of $n$ POVMs $\Q_1, \dots, \Q_n$ with instruments $\II_{\Q_1}\ldots\II_{\Q_n}$ and let us assume that 
\begin{equation}\label{eq:n-1cond}
\begin{split}
&\Q_2 \ND  \II_{\Q_3}^*\II_{\Q_4}^*\ldots \Q_n\\
&\ \ \vdots \qquad \vdots \\
&\Q_{n-1} \ND \Q_n.
\end{split}
\end{equation}
are equivalent to the NSIT condition \eqref{eq:nsit_gen} for the subsequence $\Q_2 \rightarrow \Q_3 \rightarrow \dots \Q_n$. 

Adding an additional measurement $\Q_1$ in the beginning of the sequence might introduce disturbance at any point in the sequence. For the inductive step, we need to prove that Eq.~\eqref{eq:n-1cond} together with the new condition
\begin{equation}\label{eq:1cond}
\Q_1 \ND  \II_{\Q_2}^*\II_{\Q_3}^*\ldots \Q_n,
\end{equation}
imply that the $n$-term NSIT conditions are satisfied for any state $\rho$, namely,
\begin{eqnarray}\label{eq:NSITtocheck_rec}
&p(0,{\bf q}_{\rm 2\rightarrow n}|0,{\bf s}_{\rm 2\rightarrow n})
=\sum_{q_{1}}p({\bf q}_{\rm 1\rightarrow n}|1,{\bf s}_{\rm 2\rightarrow n}) , \\
\label{eq:NSITtocheck_2} &p(q_1,0,{\bf q}_{\rm 3\rightarrow n}|1,0,{\bf s}_{\rm 3\rightarrow n})
=\sum_{q_{2}}p({\bf q}_{\rm 1\rightarrow n}|1,1,{\bf s}_{\rm 3\rightarrow n}),\  \\
\nonumber &\qquad \qquad  \qquad  \vdots \\
\nonumber & p({\bf q}_{\rm 1\rightarrow n-2},0, q_{n}|1_{\rm 1\rightarrow n-2},0,1)
\\ \label{eq:NSITtocheck_last} &\qquad \qquad =\sum_{q_{n-1}}p({\bf q}_{\rm 1\rightarrow n}|1_{\rm 1\rightarrow n}),
\end{eqnarray}
for all possible settings ${\bf s}$ and ${\bf q}$. Here $1_{\rm 1\rightarrow n}$ is a shorthand notation for $(1,1,\dots,1)$, and the probabilities are computed via the instruments $\II_{\Q_1}\ldots\II_{\Q_n}$. 

To prove that Eqs.~\eqref{eq:NSITtocheck_2}-\eqref{eq:NSITtocheck_last} are satisfied, one notices that the NSIT conditions for the sequence $2\rightarrow n$ are satisfied for any state. 
Writing the state after the first measurement as $\II_{\Q_i}^{q_1}(\rho)$ we get
\begin{equation}
p(q_1,{\bf q}_{\rm 2\rightarrow n}|1,{\bf s}_{\rm 2\rightarrow n})= \tr{\II^{{\bf q}_{\rm 2\rightarrow n}}_{\Q_{\rm 2\rightarrow n}}\circ \II^{q_1}_{\Q_1} (\rho)}.
\end{equation}
Hence, noting that the NSIT conditions do not depend on the normalisation of $\II_{\Q_i}^{q_1}(\rho)$
one sees that Eqs.~\eqref{eq:NSITtocheck_2}-\eqref{eq:NSITtocheck_last} are satisfied.

Finally, we need to prove that the NSIT conditions in Eq.~\eqref{eq:NSITtocheck_rec} are satisfied for any choice of settings ${\bf s}_{\rm 2\rightarrow n}$ and outcomes ${\bf q}_{\rm 2\rightarrow n}$ and for all initial states. The case ${\bf s}_{\rm 2\rightarrow n}= 1_{\rm 2\rightarrow n}$ is implied by Eq.~\eqref{eq:1cond}. For the case of exactly one '$0$' appearing in ${\bf s}_{\rm 2\rightarrow n}$ in position $k$, with $2\leq k\leq n$, it is sufficient to combine Eq.~\eqref{eq:1cond} with the condition $\Q_k \ND  \II_{\Q_{k+1}}^*\II_{\Q_{k+2}}^*\ldots \Q_n$ appearing in Eq.~\eqref{eq:n-1cond} and with the above reasoning on the sub-normalised initial states. This process can be iterated to get the case of more than one '$0$' appearing in the sequence of settings, which concludes the proof. \qed

{\it Proof of Obs.~\ref{obs:3}.---} The proof follows from the fact that AoT must hold for quantum measurements. In fact, given the conditions $\Q_{\pi(1)} \ND  \II_{\Q_{\pi(2)}}^*\ldots \Q_{\pi(n)}$ for all $\pi(1),\ldots,\pi(n)$, each condition of the form $\Q_k \ND  \II_{\Q_{k+1}}^*\II_{\Q_{k+2}}^*\ldots \Q_n$, associated with the fixed order $(1,\ldots,n)$, can be obtained via the condition for permutation $k,k+1,\ldots,n,1,2,\ldots,k-1$ and applying the AoT conditions on the last part of the sequence, i.e., $1,2,\ldots,k-1$. In this way, we obtain all the conditions of Eq.~\eqref{eq:fullMRcondnpoint} for the order $(1,\ldots,n)$. Repeating the argument for all orders finishes the proof.  \qed

For the proof of Obs.~\ref{obs:2wNDcomm}, we need to summarize the construction of noncommuting repeatable POVMs introduced in \cite{HW2010}. First, the dimension of the Hilbert space must be $d={\rm dim} \mathcal H \geq 5$. We see the space as a direct sum of a $3$ -dimensional and a $(d-3)$-dimensional subspace $\mathcal H_3$ and $\mathcal H_{d-3}$ as $\mathcal H=\mathcal H_3\oplus \mathcal H_{d-3}$. In $\mathcal H_3$ we fix three orthogonal one-dimensional projections $P_1,P_2,P_3$, while in $\mathcal H_{d-3}$ we fix two noncommuting projections $R_1,R_2$. Then we define an observable $\A$ as
\begin{eqnarray}\label{eq:Aobs}
A_1 = P_1 \oplus \frac{1}{2} R_1, \quad A_2 = P_2 \oplus \frac{1}{2} R_2 \, , \\
A_3 = P_3 \oplus ( \openone -\frac{1}{2} R_1 - \frac{1}{2} R_2 ) \, .
\end{eqnarray}
Since $1$ is an eigenvalue for each POVM element of $\A$, the measurement is repeatable and, in particular, it does not disturb itself. Moreover, the elements of $\A$ do not commute since $[R_1,R_2]\neq 0$. 

We can slightly modify the above construction to make the second observable in the sequence not coinciding with the first, hence, showing that the difference between 2-way nondisturbance and noncommutativity is not only a special feature of repeatable (or first kind) measurement scenarios.

{\it Proof of Obs.~\ref{obs:2wNDcomm}.---}. Given $\A$ as above and $\A'$ defined as $A'_1:=A_1$, $A'_2:=A_2+A_3$, the pair of observables $(\A,\A^\prime)$ is nondisturbing in both directions, i.e., $\A \NND \A^\prime$ and yet non commuting.
In fact, since $\A$ is repeatable, there exists $\II^\A$ such that
\begin{gather}
\sum_x (\II_{A}^{x})^* (A'_1) = A_1= A^\prime_1, \\
\sum_x (\II_{A}^{x})^*  (A'_2) = \sum_x (\mathcal I_{A}^{x})^* (A_2 + A_3)= A_2 + A_3= A'_2.
\end{gather}
In the other direction, we can simply take the instrument given by $\II_{\A^\prime}^1 := \II_{\A}^{1}$ and 
$\II_{\A^\prime}^2 := \II_{\A}^{2}+\II_{\A}^{3}$. 
Then $\sum_x (\II_{\A}^{x})^* = \sum_x (\II_{\A^\prime}^{x})^*$, hence, $\A^\prime$ does not disturb $\A$ as well. Finally, the observables are still noncommuting since $[A_2,A^\prime_1]=[A_2,A_1]\neq 0$. \qed

Starting from this example, we consider a sequence of $3$ measurements and construct a {\it hollow triangle} of POVMs with respect to the nondisturbance relations, i.e., a triple of observables such that each pair is mutually nondisturbing, but still any instrument associated with the first (according to a specific order) must disturb any nondisturbing sequence of the other two observables. This is analogous to the hollow triangle of joint measurability introduced in Ref.~\cite{Heinosaari2008, Yu2013holl}.

{\it Proof of Obs.~\ref{obs:hollow_triang}.---}. Consider a Hilbert space of dimension $d\geq 5$. One can show that there exist two observables $\A$ and $\B$ that have the following properties: (i) they are of the first kind and (ii) $\A \NND \B$ and $\A \NND \A$ holds, i.e., all the two point sequences can be implemented with nondisturbing instruments. 
However, the sequences $\A \rightarrow (\A \rightarrow \B)$ and $\A \rightarrow (\B \rightarrow \A)$ cannot be implemented with a nondisturbing instrument.
In particular, the second observable can be taken as 
\begin{eqnarray}\label{eq:b12}
B_1 = \openone \oplus 0,\ B_2 = 0 \oplus \openone ,
\end{eqnarray}
where the direct sum is w.r.t. the direct sum decomposition $\mathcal H=\mathcal H_3\oplus \mathcal H_{d-3}$.

To prove the above statements, let us consider $\A$ as in Eq.~(\ref{eq:Aobs}). First, let us prove the case $d=5$. The observable $\B$ commutes with $\A$ and, hence, the L\"uders instrument is nondisturbing in both directions. Let us consider the POVM with elements $E_{ij}=\sqrt{B_j}A_i \sqrt{B_j}=A_iB_j$, that is obtained from the sequences $(\A \rightarrow \B)$ and $(\B \rightarrow \A)$ with the L\"uders instrument. We have that $E_{11}$, $E_{21}$ and $E_{31}$ are the three projectors $P_1\oplus 0$, $P_2\oplus 0$ and $P_3\oplus 0$. The other three POVM elements, $E_{i2}$, are given by $0 \oplus R_1/2$, $0 \oplus R_2/2$ and $0 \oplus (\openone_3 - R_1/2 - R_2/2)$, respectively. 
Let us choose $R_1$ and $R_2$ as the two-dimensional projectors $(\openone + \sigma_x)/2$ and $(\openone + \sigma_z)/2$. 

According to Proposition~3 in Ref.~\cite{HW2010}, if $E_{ij}^2 \in {\rm span}\{E_{ij} |(i,j) \}$, then $\A$ is nondisturbing w.r.t. to $\E$, if and only if all their elements commute. Let us show that such operators belong to the span. The statement is trivial for $E_{i1}$ since they are all projectors, and similarly for the case $E_{12},E_{22}$, since also $R_1$ and $R_2$ are projectors. The only non trivial case is for $E_{32}$, which amounts to show that ${(\openone - R_1/2 - R_2/2)^2\in {\rm span}(R_1, R_2, \openone)}$. First, let us notice that ${\rm span}(R_1, R_2, \openone)={\rm span}(\sigma_x, \sigma_z, \openone)$. Then, we have $\left(\openone - R_1/2 - R_2/2\right)^2 = \left(\openone - \tfrac{1}{4}(\openone + \sigma_x) - \tfrac{1}{4}(\openone + \sigma_z) \right)^2$.
The only elements out of the span could be the product of $\sigma_x$ and $\sigma_z$. However, since we are taking the square, such products will appear as $\sigma_x \sigma_z + \sigma_z \sigma_x$, which is zero. 

Then $\A$ is nondisturbing w.r.t. to $\E$, if and only if all their elements commute.
However, this is not the case, since, e.g., $[A_1, E_{22}]\neq 0$. Hence, we obtain that
the sequence $\A \rightarrow \E$ cannot be implemented with a nondisturbing instrument. 
The proof for $d=5$ ends by noting that, since $\B$ is projective, the L\"uders instrument gives the unique joint measurement of $\A$ and $\B$ (see Prop. 8 in \cite{Heinosaari2008}) and, hence, any other nondisturbing implementation (if existing) results in the same joint measurement. The proof can be easily extended to the case of arbitrary odd dimension $d$, by decomposing  $\mathcal H_{d-3}=\mathcal{H}_2 \otimes \mathcal{H}_{\frac{d-3}{2}}$ and projectors $\widetilde{R}_i = R_i \otimes \openone$. Similarly for even $d$, we just use the decomposition $\mathcal H=\mathcal H_4\oplus \mathcal H_{d-4}$, increase the rank of $P_3$ and decompose $\mathcal H_{d-4}$ as the  tensor product $\mathcal{H}_2 \otimes \mathcal{H}_{\frac{d-4}{2}}$.\qed

The above is an example of a triple of POVMs that shows {\it genuine triplewise nondisturbance} in the sense that it satisfies Eq.~(\ref{NSIT4a}) for all orderings (and some instruments), but it violates Eq.~(\ref{NSIT1a}) for some ordering (and all instruments of the first measurement together with any nondisturbing instrument of the second measurement). It requires the dimension of the Hilbert space to be at least $d=5$ and two of the POVMs to be equal. The following example, instead, circumvents these limitations. However, one should note that not all of the pairs are anymore two-way nondisturbing. The example is minimal in the sense that it uses a qutrit (in qubit hollow triangles don't exist because nondisturbance is equivalent to commutativity) and binary measurements.

Consider the following triple $(\A,\B,\C)$ of POVMs 
\begin{equation}
\begin{aligned}
A_1 &= \frac 1 4 \left(\begin{matrix}2&0&-\sqrt 2 \\
0&4&0 \\
-\sqrt 2 & 0 & 3 \end{matrix}\right)  \ & \  A_2 &= \frac 1 4 \left(\begin{matrix}2&0&\sqrt 2 \\
0&0&0 \\
\sqrt 2 & 0 & 1 \end{matrix}\right) \\
B_1 &= \left(\begin{matrix}1&0&0 \\
0&0&0 \\
0 & 0 & 1 \end{matrix}\right)  \ & \  B_2 &=  \left(\begin{matrix}0&0&0 \\
0&1&0 \\
0 & 0 & 0 \end{matrix}\right) \\
C_1 &= \frac 1 {12} \left(\begin{matrix}4&0&0 \\
0&6&0 \\
0 & 0 & 5 \end{matrix}\right)  \ & \  C_2 &= \frac 1 {12} \left(\begin{matrix}8&0&0 \\
0&6&0 \\
0 & 0 & 7 \end{matrix}\right) .
\end{aligned}
\end{equation}
The sequence $\A\rightarrow \B \rightarrow \C$ is such that $\A\stackrel{ND}{\leftrightarrow}\B$, $\A\stackrel{ND}{\rightarrow}\C$ (while $\C$ must disturb $\A$) and $\B\stackrel{ND}{\leftrightarrow}\C$.
However: (i) $\A$ disturbs all the nondisturbing implementations of the sequences $\B\leftrightarrow \C$; and moreover (ii) the sequence $\B\rightarrow \A \rightarrow \C$ is nondisturbing, which implies, in particular, that the triple $(\A , \B , \C)$ is jointly measurable.

{\it Proof.---} The proof makes use of the construction in \cite{HW2010}, exploiting Remark 2 of \cite{arveson1972}.
Let us consider the nilpotent channel $\Lambda$ that maps a $3\times3$ matrix $a_{ij}$ into ${\rm diag(a_{11}, a_{22}, (a_{11}+a_{22})/2)}$. This channel is such that $\Lambda(\Lambda(X))=\Lambda(X)$ for all $3\times 3$ matrices $X$. A decomposition of $\Lambda$ into Kraus operators is given by
\begin{equation}
\begin{gathered}
K_1=\tfrac 1 2 \left(\begin{matrix} \sqrt 2&0&0 \\
0&0&0 \\
-1 & 0 & 0 \end{matrix}\right) \qquad K_2=\tfrac 1 {10} \left(\begin{matrix} 0&0&0 \\
0&-\sqrt{10}&0 \\
0&2\sqrt{10}&0 \end{matrix}\right) \\
K_3=\tfrac 1 2 \left(\begin{matrix} 0&0&0 \\
0&\sqrt 2&0 \\
0 & 0 & 0 \end{matrix}\right) \qquad K_4=\tfrac 1 {20} \left(\begin{matrix} 0&0&0 \\
0&4\sqrt{10}&0 \\
0&2\sqrt{10}&0 \end{matrix}\right) \\
K_5=\tfrac 1 2 \left(\begin{matrix} \sqrt 2&0&0 \\
0&0&0 \\
1 & 0 & 0 \end{matrix}\right) .
\end{gathered}
\end{equation}
The POVM elements of $\C$ are constructed as $C_1=\Lambda(\tfrac 1 3 P_1 + \tfrac 1 2 P_2)$ and $C_2=\Lambda(\tfrac 2 3 P_1 + \tfrac 1 2 P_2)$, where $P_1={\rm diag}(1,0,0)$ and $P_2={\rm diag}(0,1,0)$ are the projectors onto the first two computational basis elements. The POVM elements of $\A$ are constructed as $A_i=\mathcal I_i^*(\openone)$ 
where $\mathcal I_1^*(\cdot)=\sum_{j=1}^4 K_j \cdot K_j^*$ and $\mathcal I_2^*(\cdot)=K_5 \cdot K_5^*$. Thus, we have that $\sum_j \mathcal I_1^*(\cdot) =\Lambda$ and therefore $\A$ does not disturb $\C$ by construction. The rest of pairwise nondisturbance relations follow from commutativity.

However, $\A$ must disturb $\E:=\II_\B^*\C$. This can be seen as follows: we have that  
the POVM elements $\{ E_{ij} \}$ span the whole subspace of $3\times 3$ diagonal matrices. Therefore we have $E_{ij}^2 \in {\rm span} \{ E_{ij} \}$ for all $(i,j)$. This implies that $A\stackrel{ND}{\rightarrow}E$ if and only if $\A$ and $\E$ commute, which is not the case. Finally, since $\B$ is a PVM, the L\"uders instrument gives the unique joint measurement of $\B$ and $\C$ (cf. Proposition~8 in Ref.~\cite{Heinosaari2008}) and, hence, any other nondisturbing implementation (if existing) would give the same sequential POVM (i.e. the joint measurement).
Also, the fact that $\C$ must disturb $\A$ follows from $A_x^2 \in {\rm span}\{A_x\}$ for all $x$ and the fact that $\A$ and $\C$ do not commute. In particular, since $A_2$ is proportional to a projector and such that, e.g., $[C_1,A_2]\neq 0$, then it must be disturbed by every instrument implementing $\C$.
 Finally, it can be easily checked via explicit calculation that the sequence  $\B\rightarrow \A \rightarrow \C$ is nondisturbing. It is sufficient to consider as a sequential nondisturbing measurement for $\A \rightarrow \C$ the one given by the instruments $\mathcal I_1^*(\cdot)=\sum_{j=1}^4 K_j \cdot K_j^*$ and $\mathcal I_2^*(\cdot)=K_5 \cdot K_5^*$. The resulting joint POVM $F_{ij}:= \mathcal I_i^*(C_j)$, then, commutes with $B_k$ for any $i,j,k$, providing a nondisturbing measurement for the sequence $\B\rightarrow \A \rightarrow \C$. As a further consequence, the triple $(\A , \B , \C)$ is also jointly measurable . \qed

\section{Time-dependent version of Observation 1}\label{App:A}

First, we show that not all pairs of POVMs can be connected by time evolution. In fact, consider two POVMs $\Q$ and $\Q'$, such that $\Q_x$ has maximal eigenvalue $\lambda$ and $\Q'_x$ has maximal eigenvalue $\lambda' > \lambda$. Then, there is no channel $\Lambda$, i.e., a CPTP map, such that $\Lambda^*(\Q_x)=\Q'_x$, since by positivity $ \lambda \openone - \Q_x \geq 0\Rightarrow \Lambda^*(\lambda \openone - \Q_x ) \geq 0$, whereas by linearity and unitality of $\Lambda^*$,  $\Lambda^*(\lambda \openone- \Q_x )=\lambda \openone - \Q'_x \ngeq 0$. Similarly, one can prove that if the minimal eigenvalue of $\Q'_x$ is strictly smaller than the minimal eigenvalue of $\Q_x$, there is no channel mapping $\Q$ into $\Q'$. Moreover, one can relax the requirement $\Lambda^*(\Q_x)=\Q'_x$, $\forall x$, to  $\Lambda^*(\Q_x)=\Q'_{\pi(x)}$, $\forall x$, for some permutation $\pi$ of the outcomes, i.e., by allowing as a possible operation also a relabeling of outcomes, and construct a similar example. Another simple example is given by the maximally mixed POVM $\{\openone / N\,\ldots,\openone / N \}$, which is always mapped onto itself due, again, to linearity and unitality of the adjoint channel. 

Below, we show how to modify Observation 1 to cover the time-dependent case.
Given measurements $\Q_1, \Q_2$ and $\Q_3$ together with time evolutions $\Lambda_{t_1\rightarrow t_2}$ and $\Lambda_{t_2\rightarrow t_3}$ between them, we have the following observation.

\begin{lemma}
The no-signalling in time conditions on a limited set $S$ of states are equivalent to the fact that $(i)$ the first measurement doesn't disturb the sequential measurement $\Lambda_{t_1\rightarrow t_2}^*\circ\mathcal I_{\Q_2}^*\circ\Lambda_{t_2\rightarrow t_3}^*(\Q_3)$ for states in the set $S$ and that $(ii)$ the second measurement does not disturb the time evolved version of the third measurement (i.e.$\Lambda_{t_2\rightarrow t_3}^*(\Q_3)$) in the limited set of states $\left\lbrace \Lambda_{t_1\rightarrow t_2}(\rho)|\rho\in S\right\rbrace\cup\left\lbrace \frac{\Lambda_{t_1\rightarrow t_2}(\mathcal I_{\Q_1}^x(\rho))}{\tr{\Lambda_{t_1\rightarrow t_2}(\mathcal I_{\Q_1}^x(\rho))}}|\rho\in S,x\right\rbrace$. Note that here again $(i)$ includes the assumption that the instrument $\mathcal I_{\Q_2}$ fulfills $(ii)$.
\end{lemma}

\begin{proof}
Assume that the nondisturbance conditions hold. The observed three point probability distribution reads
\begin{align*}
p(x,y,z|1,1,1)=\tr{\Lambda_{t_2\rightarrow t_3}\circ\mathcal I_{\Q_2}^y\circ\Lambda_{t_1\rightarrow t_2}\circ\mathcal I_{\Q_1}^x(\rho)Q_3(z)}.
\end{align*}

For this scenario there are three NSIT conditions to be checked. As the first measurement does not disturb the sequential measurement of the second and the third measurement one gets the first NSIT condition
\begin{align}
\sum_x &p(x,y,z|1,1,1)=\\
&=\sum_x \tr{\rho\mathcal I_{\Q_1}^{x*}\circ\Lambda_{t_1\rightarrow t_2}^*\circ\mathcal I_{\Q_2}^{y*}\circ\Lambda_{t_2\rightarrow t_3}^*(Q_3(z))}\\
&=\tr{\rho\Lambda_{t_1\rightarrow t_2}^*\circ\mathcal I_{\Q_2}^{y*}\circ\Lambda_{t_2\rightarrow t_3}^*(Q_3(z))}\\
&=p(0,y,z|0,1,1).
\end{align}

Using that the second measurement does not disturb the third measurement on the set $\left\lbrace \frac{\Lambda_{t_1\rightarrow t_2}(\mathcal I_{\Q_1}^x(\rho))}{\tr{\Lambda_{t_1\rightarrow t_2}(\mathcal I_{\Q_1}^x(\rho))}}|\rho\in S,x\right\rbrace$ of states we get the second condition
\begin{align}
\sum_y &p(x,y,z|1,1,1)=\\
&=\sum_y \tr{\rho\mathcal I_{\Q_1}^{x*}\circ\Lambda_{t_1\rightarrow t_2}^*\circ\mathcal I_{\Q_2}^{y*}\circ\Lambda_{t_2\rightarrow t_3}^*(Q_3(z))}\\
&=\tr{\rho\mathcal I_{\Q_1}^{x*}\circ\Lambda_{t_1\rightarrow t_2}^*\circ\Lambda_{t_2\rightarrow t_3}^*(Q_3(z))}\\
&=p(x,0,z|1,0,1).
\end{align}

The last NSIT condition follows directly from the fact that the second measurement does not disturb the third one on the set $\left\lbrace \Lambda_{t_1\rightarrow t_2}(\rho)|\rho\in S\right\rbrace$ of states.

For the other direction of the proof, one can show that the violation of either one of the nondisturbance conditions leads to the violation of the NSIT conditions. First, assume that the second measurement disturbs the third for some state in the set $\{\Lambda_{t_1\rightarrow t_2}(\rho)|\rho\in S\}$. By writing
\begin{align}
\sum_y & p(0,y,z|0,1,1)=\\
&=\sum_y\tr{\Lambda_{t_1\rightarrow t_2}(\rho)\mathcal I_{\Q_2}^{y*}\circ\Lambda_{t_2\rightarrow t_3}^*(Q_3(z))}\\
&\neq\tr{\Lambda_{t_1\rightarrow t_2}(\rho)\Lambda_{t_2\rightarrow t_3}^*(Q_3(z))}\\
&=p(0,0,z|0,0,1)
\end{align}
one sees that the third NSIT condition does not hold for some $\rho\in S$ and some $z$.

Second, assume that the second measurement is disturbing the third one for some state in the set 
$\left\lbrace \frac{\Lambda_{t_1\rightarrow t_2}(\mathcal I_{\Q_1}^x(\rho))}{\tr{\Lambda_{t_1\rightarrow t_2}(\mathcal I_{\Q_1}^x(\rho))}}|\rho\in S,x\right\rbrace$ of states. Then
\begin{align}
\sum_y &p(x,y,z|1,1,1)=\\
&=\sum_y \tr{\rho\mathcal I_{\Q_1}^{x*}\circ\Lambda_{t_1\rightarrow t_2}^*\circ\mathcal I_{\Q_2}^{y*}\circ\Lambda_{t_2\rightarrow t_3}^*(Q_3(z))}\\
&\neq\tr{\rho\mathcal I_{\Q_1}^{x*}\circ\Lambda_{t_1\rightarrow t_2}^*\circ\Lambda_{t_2\rightarrow t_3}^*(Q_3(z))}\\
&=p(x,0,z|1,0,1)
\end{align}
for some $\rho\in S$ and some $x,z$.

Finally, for the case of the first measurement disturbing the sequential measurement of the second and the third measurement we check the first NSIT condition
\begin{align}
\sum_x &p(x,y,z|1,1,1)=
\\ &=\sum_x\tr{\rho\mathcal I_{\Q_1}^{x*}\circ\Lambda_{t_1\rightarrow t_2}^*\circ\mathcal I_{\Q_2}^{y*}\circ\Lambda_{t_2\rightarrow t_3}^*(Q_3(z))}\\
&\neq\tr{\rho\Lambda_{t_1\rightarrow t_2}^*\circ\mathcal I_{\Q_2}^{y*}\circ\Lambda_{t_2\rightarrow t_3}^*(Q_3(z))}\\
&=p(0,y,z|0,1,1)
\end{align}
for some $\rho\in S$ and some $y,z$.
\end{proof}

The time-dependent version of Observation 1 follows directly: the above lemma is formulated for a limited set of states. When applied to the whole state space, one sees that $\left\lbrace\frac{\Lambda_{t_1\rightarrow t_2}(\mathcal I_{\Q_1}^x(\rho))}{\tr{\Lambda_{t_1\rightarrow t_2}(\mathcal I_{\Q_1}^x(\rho))}}|\rho\in\mathcal S(\mathcal H),x\right\rbrace\subseteq\{\Lambda_{t_1\rightarrow t_2}(\rho)|\rho\in\mathcal S(\mathcal H)\}$. Hence, we have

\begin{obs}\label{obs:6}
For time-dependent scenarios with optimisation over states, macrorealism is equivalent to the first measurement not disturbing the sequential measurement $\Lambda_{t_1\rightarrow t_2}^*\circ\mathcal I_{\Q_2}^*\circ\Lambda_{t_2\rightarrow t_3}^*(\Q_3)$ (for all states) and the second measurement not disturbing the time evolved version of the third one (i.e.$\Lambda_{t_2\rightarrow t_3}^*(\Q_3)$) within the time evolved set $\{\Lambda_{t_1\rightarrow t_2}(\rho)|\rho\in\mathcal S(\mathcal H)\}$ of states.
\end{obs}

\section{Disturbance measure}\label{App:C}

\subsection{Semidefinite programming definition of the disturbance measure}\label{app:semidef}
For finite dimensional Hilbert spaces, an SDP definition of the disturbance measure can be defined as follows.
First, let us recall the Choi-Jamio\l{}kowski isomorphism~\cite{Jamiolkowski,Choi}, namely,
\begin{equation}
\begin{split}
(\II_\A^x)^*(B_y) = {\text{tr}_2}[M_\A^x\  \openone \otimes B_y],\\
\text{ with } M_\A^x :=  \left[ ({\rm Id} \otimes \II_\A^x) \ketbra{\Omega} \right]^t,
\end{split}
\end{equation}
where $\ket{\Omega}$ is the (unnormalized) maximally entangled state $\ket{\Omega} = \sum_i \ket{ii}$, the transposition $^t$ is taken w.r.t. the $\{\ket{ij} \}_{ij}$ basis, and ${\rm Id}$ represents the identity channel. Using that the norm of an Hermitian operator $X$ can be written as $\|X\| = \min \{ \lambda\ | -\lambda \openone \leq X \leq \lambda \openone\}$, we can formulate  $D_{\A}(\B)$ as the following SDP:
\begin{equation}
\begin{split}
\min_{\{\lambda_y\}_y, \{M_\A^x\}_x}  &\sum_y \lambda_y \\
\text{subject to: }& -\lambda_y \openone \leq  {\text{tr}_2}[\sum_x M_\A^x \openone \otimes B_y] - B_y\leq \lambda_y \openone,\\
&\text{ for all } y,\\
& M_\A^x \geq 0, \text{ for all } x\\
& \text{tr}_2[M_\A^x]=A_x \text{ for all } x .
\end{split}
\end{equation}

Contrary to the case of two measurements, in the case of three (or more) measurements the evaluation of the measures does not seem like an SDP anymore. Nevertheless, one can still try to perform some numerical optimization using SDP. In fact, the function we want to minimize is linear in each argument $\II_{\Q_i}$ (or, more precisely, in its Choi matrices $\{M_{\Q_i}^{q_i}\}_{q_i}$), whenever the other instruments are kept fixed. This gives an SDP analogous to that of Eq.~\eqref{eq:DAB_def}. The so-called see-saw method, then, consists in iterating the maximization over each argument $\{M_{\Q_i}^{q_i}\}_{q_i}$, keeping at each step the found optimal solution. Even though the convergence to the optimal solution cannot be guaranteed, this method has been widely applied to problems in nonlocality and steering \cite{Werner2001See-saw, YCLPRA2007See-saw, PalPRA2010, CavalcantiReview2017}.

\subsection{Free operations for macrorealism}
The first example of free operations are classical post-processings of POVMs. Recall that a (classically) post-processed version $\td\A$ of a POVM $\A$ is defined through the POVM elements $\td A^a=\sum_{a'}p(a|a')A_{a'}$, where $p(\cdot|a')$ is a probability distribution for every $a'$. To be more concrete, we wish to concentrate on the three measurement setting and carry out the calculation using a term from the sum in Eq.~(\ref{Macromeas3}) with a trivial permutation. Note that any instrument $\II_{\Q_i}$ of $\Q_i$ can be post-processed into an instrument of $\td\Q_i$ as $\II_{\td\Q_i}^a=\sum_{a'}p(a|a')\II_{\Q_i}^{a'}$. Using post-processed versions $\{\td\Q_i\}_{i=2}^3$ of the involved POVMs and the respective post-processed instruments one has by triangle inequality
\begin{align}
&\sum_{y,z} \Big\|(\Lambda_{\mathcal I_{\td\Q_1}}^*-{\rm Id})(\II_{\td \Q_2}^{y*}(\td Q_3^z))\Big\|\nonumber\\
\leq&\sum_{y,z} \Big\|(\Lambda_{\mathcal I_{\Q_1}}^*-{\rm Id})(\II_{\Q_2}^{y*}(Q_3^z))\Big\|,
\end{align}
where ${\rm Id}$ is the identity map. On the left hand side we have used an instrument that originates from the one on the right hand side. Hence, it is clear that when taking the infimum on the right hand side, we cover some (but not necessarily all) possible instruments on the left hand side. Taking the infimum on the left hand side can only decrease the value and, hence, we conclude that classical post-processing is a free operation for $i=2,3$. For the case $i=1$ we refer to Corollary 2 in Ref.~\cite{TeikoJPA2017}, where it is shown that the set of total channels (of instruments) associated to a given POVM is a subset of the set of total channels (of instruments) associated to a post-processed version of this POVM. Hence, the infimum over instruments results in a (possibly) larger set of total channels for post-processed POVMs, than for the original ones. We conclude that 
\begin{align}
{\rm MR}(\{\td \Q_i\}_{i=1}^3)\leq {\rm MR}(\{\Q_i\}_{i=1}^3).
\end{align}

Let us consider global preprocessing, i.e., a quantum channel applied to all POVMs in a sequence of measurements. Define $\td\Q_1= \EE^*(\Q_1)$ with elements $\td Q_1^x=\EE^*(Q_1^x)$.  Note that we can compose the channel $\EE$ with any instrument $\II_{\Q_i}$ of $\Q_i$ and obtain another instrument $\II^\prime_{\Q_i}:=\EE \circ \II_{\Q_i}$ of $\Q_i$. Furthermore, we can compose any instrument $\II_{\Q_1}$ of $\Q_i$ with the channel $\EE$ and obtain an instrument $\II_{\td\Q_1}:=\II_{\Q_1} \circ \EE$ implementing $\td\Q_1$. 

As an example, consider the case of three measurements $\Q_1,\Q_2,\Q_3$, the preprocessed measurements will be $\td \Q_i := \EE (\Q_i)$. Given some instruments $\II_{\Q_i}$ for the original observables, we consider new instruments $\II_{\td\Q_i}:=\II_{\Q_i} \circ \EE$. We have
\begin{align}\label{eq:pre-proc}
&\sum_{y,z} \Big\|(\Lambda_{\mathcal I_{\td\Q_1}}^*-{\rm Id})(\II_{\td \Q_2}^{y*}(\td Q_3^z))\Big\| \nonumber\\
=&\sum_{y,z} \Big\|(\EE^*\circ \Lambda_{\mathcal I_{\Q_1}}^*-{\rm Id})(\EE^*\circ \II_{ \Q_2}^{y*}(\EE^*( Q_3^z)))\Big\|\nonumber\\
\leq&\sum_{y,z} \Big\|(\Lambda_{\mathcal I_{\Q_1}}^*\circ \EE^*-{\rm Id})(\II_{ \Q_2}^{y*}\circ\EE^* (Q_3^z))\Big\|,
\end{align}
where we used that the leftmost $\EE^*$ can be simply consider to be applied to the state, over which we maximize. Notice that $\II_{ \Q_i}^{*}\circ\EE^*$ is a valid instrument for $\Q_i$, however, it will not be, in general, the optimal one, i.e., the one that minimizes the disturbance. One possibility, is to consider an invertible channel $\EE$, the optimal instruments $\II_{\Q_i}^{\rm opt}$ minimizing $D_{\Q_1}( \Q_2, \Q_3, \II_{\Q_1},\II_{ \Q_2})$, and define $\II_{\Q_i} := \EE^{-1}\circ \II_{\Q_i}^{\rm opt}$, which is still a valid instrument for $\Q_i$. By substituting in Eq.~\eqref{eq:pre-proc}, we obtain
\begin{align}\label{eq:pre-proc2}
&\sum_{y,z} \Big\|(\Lambda_{\mathcal I_{\Q_1}}^*\circ \EE^*-{\rm Id})(\II_{ \Q_2}^{y*}\circ\EE^* (Q_3^z))\Big\|\nonumber\\
=&\sum_{y,z} \Big\|(\Lambda_{\II_{\Q_1}^{\rm opt}}^*-{\rm Id})((\II_{ \Q_2}^{\rm opt})^{y*} (Q_3^z))\Big\|\nonumber\\
=& \inf_{\II_{\Q_i},\II_{\Q_2}} D_{\Q_1}( \Q_2, \Q_3, \II_{\Q_1},\II_{ \Q_2}).
\end{align}
The same argument can be straightforwardly generalized to longer sequences. However, this result could have been obtained straightforwardly by noticing that we are restricted to mappings of the space into itself, hence, the only invertible operations are unitary transformations. It is not obvious how to extend it to arbitrary channels.

For the special case of qubit channels, we can prove that every preprocessing is a free operation. This result is based on the notion of {\it commutativity preserving channel} introduced in Ref.~\cite{HuPRA2012}, where it was proven that any qubit map $\Lambda$ that is unital, i.e., $\Lambda(\openone)=\openone$, preserve the commutativity of Hermitian operators, namely $[A, B]=0$ implies $[\Lambda(A),\Lambda(B)]=0$ for any $A,B$. The argument is very simple and we, thus, summarize it here.  Since $[A, B]=0$, they have a commong eigenbases, hence the condition $[\Lambda(A),\Lambda(B)]=0$, for all $A,B$, is equivalent to the condition $[\Lambda(\ketbra{\phi}),\Lambda(\ketbra{\psi})]=0$, for all $\ket{\phi},\ket{\psi}$ with $\mean{\phi |\psi}=0$. By linearity, the latter can be written as
$\frac{1}{2}[\Lambda( \ketbra{\phi}) + \Lambda(\ketbra{\psi}), \Lambda( \ketbra{\phi}) - \Lambda(\ketbra{\psi})]
= \frac{1}{2}[\Lambda(\openone), \Lambda( \ketbra{\phi}) - \Lambda(\ketbra{\psi})]$, 
hence it is always zero if $\Lambda(\openone)=\openone$. To conclude, it is sufficient to notice that qubit POVMs are nondisturbing if and only if their elements commute~\cite{HW2010}.

A natural question is what happens if one has local, instead of global, preprocessing, namely, the application of the channel to a single measurement in the sequence. Clearly, a general local preprocessing will not be a free operation, e.g., one can rotate just one of a pair of commuting projective measurements and make them mutually disturbing.
In the following, we will show that local preprocessing via the the depolarizing channel, $\EE_{\rm dep}^*(Q_1^x)=\alpha Q_1^x + \tfrac{1-\alpha} d \tr{Q_1^x} \openone$, is a free operation.
Let us consider the measures of disturbance $D_{\td\Q_1}(\Q_2)$ and $D_{\Q_2}(\td\Q_1)$ separately.
For the latter quantity we immediately have that $D_{\Q_2}(\td\Q_1)\leq D_{\Q_2}(\Q_1)$ as $\big\|(\Lambda_{\II_{\Q_2}}^* - {\rm Id})(\td Q_1^y) \big\|\leq \alpha \big\| (\Lambda_{\II_{\Q_2}}^* - {\rm Id})(Q_1^y) \big\|$ follows from the unitality condition $\Lambda_{\II_{\Q_2}}^*(\openone) = \openone$, valid for all channels (TP property).
For the other term, i.e., $\big\| (\Lambda_{\II_{\td\Q_1}}^*-{\rm Id})Q_2^y\big\|$, we notice that any instrument of the form $\II_{\td\Q_1}^x=\alpha \II_{\Q_1}^x +\frac{1-\alpha}{d}\tr{Q_1^x}\Lambda_D$, for any choice of a channel $\Lambda_D$,  will be a valid instrument for $\td\Q_1$ as again $\Lambda^*_D(\openone)=\openone$. By taking, in particular, 
$\Lambda_D=\id$, we have
\begin{equation}
\begin{gathered}
\big\| (\Lambda_{\II_{\td\Q_1}}^*-{\rm Id})(Q_2^y) \big\|= \big\| \alpha \Lambda_{\II_{\Q_1}}^*(Q_2^y)+  (1-\alpha) Q_2^y - Q_2^y \big\| \\
 = \alpha \big\| (\Lambda_{\II_{\Q_1}}^*-{\rm Id}) (Q_2^y) \big\|  ,
\end{gathered}
\end{equation}
and thus we have 
$D_{\td\Q_1}(\Q_2) \leq D_{\Q_1}(\Q_2)$ since $|\alpha| \leq 1$. 

This argument can be extended to arbitrary sequences of POVMs. Let us consider now a sequence of three POVMs, with a fixed ordering $\Q_1 \rightarrow \Q_2 \rightarrow \Q_3$. If the depolarizing channel is applied to $\Q_1$ or $\Q_3$, the reasoning is identical to the case of two measurements above. Let us then discuss the case of a depolarizing channel applied to $\Q_2$ and define $\td Q_2^y :=\EE_{\rm dep}^*(Q_2^y)=\alpha Q_2^y + \tfrac{1-\alpha} d \tr{Q_2^y} \openone$. We define an instrument for $\Q_2$ like above as $\II_{\td\Q_2}^y=\alpha \II_{\Q_2}^y +\frac{1-\alpha}{d}\tr{Q_2^y}\Lambda_D$, with $\Lambda_D$ being any channel. To have a more compact notation, we define $\II_D^y:=\frac{1}{d}\tr{Q_2^y}\Lambda_D$.

We can then write the disturbance quantifier for this fixed sequence as the infimum over ${\II_{\Q_1},\II_{\td \Q_2}}$ of the expression
\begin{equation}
D_{\td \Q_2}(\Q_3, \II_{\td \Q_2}) + D_{\Q_1}(\td \Q_2, \Q_3, \II_{\Q_1},\II_{\td \Q_2}).
\end{equation}
By substituting the definition of $\td\Q_2, \II_{\td\Q_2}$, we have
\begin{widetext}
\begin{equation}
\begin{split}
\sum_z \| \alpha ( \Lambda_{\Q_2}^* - \id) Q_3^z + (1-\alpha) ( \Lambda_D^* - \id) Q_3^z\| 
+ \sum_{yz}  \| \alpha(\Lambda_{\Q_1}^* - \id) (\II_{\Q_2}^* \Q_3)_{yz} +  (1-\alpha) (\Lambda_{\Q_1}^* - \id) (\II_{D}^*\Q_3)_{yz}\| \\
\leq \alpha \left[ D_{ \Q_2}(\Q_3, \II_{ \Q_2}) + D_{\Q_1}( \Q_2, \Q_3, \II_{\Q_1},\II_{ \Q_2}) \right]+ (1-\alpha) \left[ \sum_z \| ( \Lambda_D^* - \id) Q_3^z\| 
+ \sum_{yz}  \| (\Lambda_{\Q_1}^* - \id) (\II_{D}^*\Q_3)_{yz}\| \right]\\
= \alpha \left[ D_{ \Q_2}(\Q_3, \II_{ \Q_2}) + D_{\Q_1}( \Q_2, \Q_3, \II_{\Q_1},\II_{ \Q_2}) \right]+ (1-\alpha) \left[ \sum_z \| ( \Lambda_D^* - \id) Q_3^z\| 
+ \sum_{z}  \| (\Lambda_{\Q_1}^* - \id) \Lambda_D^* Q_3^{z}\| \right].
\end{split}
\end{equation}
\end{widetext}
To conclude the proof, it is sufficient to choose $\Lambda_D = \Lambda_{\Q_2}$ and to notice that 
$D_{\Q_1}(\Q_2, \Q_3, \II_{\Q_1},\II_{\Q_2}) \geq D_{\Q_1}( \Lambda_{\Q_2}^*\Q_3, \II_{\Q_1})$. In fact,
\begin{equation}\label{eq:ine_tot_inst}
\begin{split}
 D_{\Q_1}( \Lambda_{\Q_2}^*\Q_3, \II_{\Q_1}) =\sum_{z}  \| (\Lambda_{\Q_1}^* - \id)\sum_y (\II_{\Q_2}^* \Q_3)_{yz}\|\\
 \leq  \sum_{yz}  \| (\Lambda_{\Q_1}^* - \id) (\II_{\Q_2}^* \Q_3)_{yz}\| = D_{\Q_1}(\Q_2, \Q_3, \II_{\Q_1},\II_{\Q_2}).
\end{split}
\end{equation}

Since we proved all three cases, and the minimization over instruments is done separately for each permutation, we can conclude that.
\begin{equation}
{\rm MR}(\td \Q_1, \Q_2, \Q_3)\leq {\rm MR}(\Q_1, \Q_2, \Q_3) .
\end{equation}
Clearly, since the measure is symmetric under exchange of its arguments, the same relation holds if we apply the depolarizing channel to $\Q_2$ or $\Q_3$.

It is clear that such an argument can be extended to longer sequences. Let us consider a sequence of $n$ measurements $\Q_1,\ldots,\Q_n$, where the depolarizing channel is applied to some measurement $s$, i.e., ${\Q_s\mapsto \td\Q_s}$. By using the same choices of instruments for $\td\Q_s$ as above, i.e., $\Lambda_D = \Lambda_{\Q_s}$, and for $1<k<s<n$, we have
\begin{widetext}
\begin{align}
\nonumber\sum_{q_{k},\ldots,q_n} \| (\Lambda_{\Q_{k-1}}^*-\id) (\II_{\Q_k}^*\cdots\II_{\td\Q_s}^*\cdots\II_{\Q_{n-1}}^*\Q_n)_{q_k\ldots,q_n}\|\\
\nonumber \leq \alpha  \sum_{q_{k},\ldots,q_n} \| (\Lambda_{\Q_{k-1}}^*-\id) (\II_{\Q_k}^*\cdots\II_{\Q_s}^*\cdots\II_{\Q_{n-1}}^*\Q_n)_{q_k\ldots,q_n}\|  \\
\nonumber + (1-\alpha) \sum_{q_{k},\ldots,q_n} \| (\Lambda_{\Q_{k-1}}^*-\id) (\II_{\Q_k}^*\cdots\II_{D}^*\cdots\II_{\Q_{n-1}}^*\Q_n)_{q_k\ldots,q_n}\| \\
\nonumber =  \alpha  \sum_{q_{k},\ldots,q_n} \| (\Lambda_{\Q_{k-1}}^*-\id) (\II_{\Q_k}^*\cdots\II_{\Q_s}^*\cdots\II_{\Q_{n-1}}^*\Q_n)_{q_k\ldots,q_n}\|  \\
\nonumber + (1-\alpha) \sum_{q_{k},\ldots,\cancel{q_s},\ldots,q_n} \| (\Lambda_{\Q_{k-1}}^*-\id) (\II_{\Q_k}^*\cdots\Lambda_{\Q_s}^*\cdots\II_{\Q_{n-1}}^*\Q_n)_{q_k \ldots \cancel{q_s}\ldots  q_n}\|\\
 \leq \sum_{q_{k},\ldots,q_n} \| (\Lambda_{\Q_{k-1}}^*-\id) (\II_{\Q_k}^*\cdots\II_{\Q_s}^*\cdots\II_{\Q_{n-1}}^*\Q_n)_{q_k \ldots q_n}\|,
\end{align}
\end{widetext}
where we used the same argument as in Eq.~\eqref{eq:ine_tot_inst} for the last inequality. The case $s=k$ is identical to the above, and similarly for the case $s=1$ and $s=n$. 

As a consequence, we can conclude that the application of a depolarizing channel is a free operation with respect to the measure MR, independently of where it is applied.

\bibliography{ND}{}

\end{document}